\documentclass[12pt, a4paper]{article}
\usepackage[top=1in, bottom=1in, left=1.25in, right=1.25in]{geometry}
\usepackage{booktabs}
\usepackage{multirow}
\usepackage{setspace}
\usepackage{natbib}
\usepackage{authblk}
\usepackage{setspace}
\usepackage[pdftex]{graphicx}
\usepackage{amsmath}
\usepackage{amsthm}
\usepackage{amsfonts}
\usepackage{ascmac}
\usepackage{amssymb}
\usepackage{bm}
\usepackage{color}
\usepackage{enumitem}
\usepackage{type1cm} 
\usepackage{mathrsfs}
\usepackage{hyperref}
\usepackage{comment}

\newtheorem{theorem}{Theorem}

\newtheorem{remark}{Remark}%

\usepackage{comment}
\usepackage{enumitem}
\usepackage{multicol}

\newcommand{\biblist}{\begin{list}{}
{\listparindent 0.0cm \leftmargin 0.50cm \itemindent -0.50 cm
\labelwidth 0 cm \labelsep 0.50 cm
\usecounter{list}}\clubpenalty4000\widowpenalty4000}
\newcommand{\ebiblist}{\end{list}}

\theoremstyle{plain}
\newtheorem{lem}{Lemma}

\title{\bf Generalized Cram\'er's coefficient via  $f$-divergence for contingency tables}
\author[1]{Wataru Urasaki}
\author[1]{Tomoyuki Nakagawa}
\author[1]{Tomotaka Momozaki}
\author[1,2]{Sadao Tomizawa}

\affil[1]{Department of Information Sciences, Tokyo University of Science}
\affil[2]{Department of Information Science, Meisei University}
\date{}

\begin{document}

\clearpage
\pagenumbering{arabic}
\maketitle

\begin{abstract}
Various measures in two-way contingency table analysis have been proposed to express the strength of association between row and column variables in contingency tables.
Tomizawa et al. (2004) proposed more general measures, including Cram\'er's coefficient, using the power-divergence. 
In this paper, we propose measures using the $f$-divergence that has a wider class than the power-divergence.
Unlike statistical hypothesis tests, these measures provide quantification of the association structure in contingency tables.
The contribution of our study is proving that a measure applying a function that satisfies the condition of the $f$-divergence has desirable properties for measuring the strength of association in contingency tables.
With this contribution, we can easily construct a new measure using a divergence that has essential properties for the analyst. For example, we conducted numerical experiments with a measure applying the $\theta$-divergence.
Furthermore, we can give further interpretation of the association between the row and column variables in the contingency table, which could not be obtained with the conventional one.
We also show a relationship between our proposed measures and the correlation coefficient in the bivariate normal distribution of latent variables in the contingency tables.
\end{abstract}

\medskip


\medskip


\section{Introduction}\label{sec1}
Contingency tables and their analysis are important for various fields, such as medicine, psychology, education, and social science. 
Typically, contingency tables are used to evaluate whether row and column variables are statistically independent.
If the independence of the two variables is rejected, for example, through Pearson's chi-squared test, or if they are clearly related, then we are interested in their strength of association.
Many coefficients have been proposed to measure the strength of association between the two variables, namely, to measure the degree of departure from independence. 
Pearson's coefficient $\phi^2$ of mean square contingency and $P$ of contingency and Tschuprow's coefficient $T$ (\citealp{Tschuprow1926grundbegriffe,Tschuprow1939principles}) serve as prime examples (see, e.g., \citealp{bishop2007discrete, everitt1992analysis, agresti2003categorical}). 
These measures can represent the strength of the association within an interval of $0$--$1$, where the value $0$ indicates the independence of the contingency table.
However, the problem with $\phi^2$ is that its value does not attain $1$ despite that the contingency table has a complete association structure (i.e., maximum departure from independence). 
Similarly, $P$ and $T$ do not always attain the value of $1$ depending on the number of rows and columns in the table.
To address this issue, \cite{cramer1946mathematical} proposed Cram\'er's coefficient $V^2$, which can reach the value of $1$ if the contingency table has a complete association structure for all rows and columns.
Specifically, $V^2$ indicates the strength of association in the contingency table as $0 \leq V^2 \leq 1$, with the value of $0$ identifying the independent structure and the value of $1$ identifying the complete association structure.

\cite{renyi1961measures} introduced a class of measures of a divergence of two distributions. 
Recent studies have linked contingency tables and the divergence. \cite{tomizawa2004generalization} proposed measures $V^2_{t(\lambda)}$ ($t=1, 2, 3$) based on the power-divergence with parameter $\lambda \geq 0$.
This study extended the measure that had been limited to examining $V^2$ ($\lambda = 1$) and showed to be the members of a single-parameter family, including a measure based on the KL-divergence ($\lambda = 0$).
(For more details of the power-divergence, see \cite{cressie1984multinomial}, and \cite{read1988goodness}.)
Furthermore, the $f$-divergence is introduced by \cite{ali1966general} and \cite{csiszar1963ene} as a useful generalization of the relative entropy, which retains some of its major properties. 
It is also called the $\phi$-divergence. 
In contingency table analysis, a  considerable amount of literature has been published on modeling using the $f$-divergence (e.g., \citealp{kateri1994f, kateri1997asymmetry, kateri2007class, fujisawa2020asymmetry, tahata2022advances, yoshimoto2019quasi}). 
Many studies on goodness-of-fit tests using the $f$-divergence have been conducted in the literature, showing the usefulness of the $f$-divergence \citep[e.g.,][etc.]{pardo2018statistical, felipe2014phi, felipe2018statistical}.
However, discussions on the measures using the $f$-divergence are limited.

In this paper, we propose a wider class of measures than the conventional one via the $f$-divergence.
This study's contribution is proving that a measure applying a function $f(x)$ that satisfies the condition of the $f$-divergence has desirable properties for measuring the strength of association in contingency tables.
This contribution allows us to easily construct a new measure using a divergence that has desirable properties for the analyst. For example, we conduct numerical experiments with a measure applying the $\theta$-divergence.
Furthermore, we can give further interpretation of the association between rows and columns in the contingency table, which could not be obtained with the conventional one.

The rest of this paper is organized as follows. 
Section \ref{sec2} proposes new measures to express the strength of the association between the row and column variables in two-way contingency tables.
Furthermore, the section shows that the proposed measures have desirable properties for measuring the strength of association.
Section \ref{sec3} presents the relationship between the measures and correlation coefficient in the bivariate normal distribution of the latent variables in the contingency tables. Section \ref{sec4} demonstrates its simulation experiment.
Section \ref{sec5} presents the approximate confidence intervals of the proposed measures. 
Section \ref{sec6} presents the analysis examples of the proposed measures applying the power-divergence and the $\theta$-divergence with actual data.
Finally, Section $7$ provides some concluding remarks.

\section{Generalized measure}\label{sec2}
We consider measures of association using the $f$-divergence for an $r \times c$ contingency table. 
For the $r \times c$ contingency table, let $p_{ij}$ denote the probability that an observation will fall in the $i$th row and $j$th column of the table $(i = 1, \dots, r; j=1, \dots, c)$.
Moreover, let $p_{i\cdot}$ and $p_{\cdot j}$ be $p_{i \cdot} = \sum^c_{t=1} p_{it} $ and $p_{\cdot j} = \sum^r_{s=1} p_{sj}$.
Hereinafter, we assume that $\{p_{i\cdot} \ne 0,$ $p_{\cdot j} \geq 0\}$ when $r \leq c$ and $\{p_{i\cdot} \geq 0,$ $p_{\cdot j} \ne 0\}$ when $r > c$. 

In \cite{Sason2016divergence}, the $f$-divergence from $P$ to $Q$ is defined as 
$I_f(P;Q) = \int f(dP/dQ) dQ$, where $f$ is a convex function and $P \ll Q$. 
For the $r \times c$ contingency table, $P$ and $Q$ are given as discrete distributions $\{p_{ij}\}$ and $\{q_{ij}\}$. Accordingly, we have $dP/dQ = \{p_{ij}/q_{ij}\}$. 
Thus, the $f$-divergence from $\{p_{ij}\}$ to $\{q_{ij}\}$ is given as
\begin{align*}
I_f(P;Q) = I_f(\{p_{ij}\};\{q_{ij}\}) &= \sum_{i}\sum_{j} q_{ij} f\left(\frac{p_{ij}}{ q_{ij}}  \right), 
\end{align*}
where $f(x)$ is a once-differentiable and strictly convex function on $(0, +\infty)$ with $f(1) = 0$, $ \lim_{x \to 0}f(x) = 0$, $0f(0/0) = 0$, and $0f(a/0) = a\lim_{x \to \infty}f(x) / x$ (see \citealp{csiszar2004information}). 
By choosing function $f$, many important divergences, such as the KL-divergence ($f(x) = x\log x$), the Pearson's divergence ($f(x) = x^2-x$),  the power-divergence ($f(x)=(x^{\lambda+1}-x)/\lambda(\lambda+1)$), and the $\theta$-divergence ($f(x) = (x-1)^2/(\theta x + 1 - \theta) + (x-1)/(1 - \theta)$ ), are included in special cases of the $f$-divergence \citep[see,][e.g.,]{Sason2016divergence, ichimori2013inequalities}. 
Furthermore, the $f$-divergence is one of the monotone and regular divergences. 
The class of monotone and regular divergences is introduced in \cite{cencov2000statistical} and studied in \cite{corcuera1998characterization} as a wide class of invariant divergences with respect to Markov embeddings. 
The class of monotone and regular divergence is often used as the measures of goodness of prediction \citep[see][etc.]{gkisser1993predictive, corcuera1999relationship, corcuera1999generalized}. 
Studying these measures aims to obtain a quantitative measure of how well a row or column variable predicts the other variable. 
Therefore, we consider that the measures using the $f$-divergence are appropriate for measuring the association and a natural generalization of that of \cite{tomizawa2004generalization}. 

Measures that present the strength of association between row and column variables are proposed in three cases: 
(I) When the row and column variables are response and explanatory variables, respectively 
(II) When they are explanatory and response variables, respectively
(III) When response and explanatory variables are undefined 
Further, we define measures for the asymmetric situation (in the case of (I) and (II)) and for the symmetric situation (in the case of (III)).

The following are the three properties that should be possessed by the measures:
(i) The measures are contained within an interval (e.g., from $0$ to $1$).
(ii) When the measure is minimal, the row and column variables are statistically independent.
(iii) When the measure is maximal, the categories of one variable can be identified from the other.
Conventional measures satisfy all of these properties.
In the remainder of this section, we prove that the proposed measures also satisfy these three properties.

\subsection{Case I}
For a asymmetric situation wherein the column variable is the explanatory variable and the row variable is the response variable, we propose the following measure that presents the strength of association between the row and column variables by
\begin{align*}
V^2_{1(f)} &= \frac{I_f(\{p_{ij}\};\{p_{i \cdot}p_{\cdot j}\})}{K_{1(f)}},
\end{align*}
where
\begin{align*}
I_f(\{p_{ij}\};\{p_{i \cdot}p_{\cdot j}\}) &= \sum^r_{i=1} \sum^c_{j=1}p_{i \cdot}p_{\cdot j} f\left(\frac{p_{ij}}{p_{i \cdot}p_{\cdot j}}  \right), \\
K_{1(f)} &= \sum^r_{i=1}p^2_{i \cdot} f\left( \frac{1}{p_{i \cdot}} \right).
\end{align*}
Then, the following theorem for the measure $V^2_{1(f)}$ is obtained.

\begin{theorem}\label{thm1}
For each convex function $f$, 
\begin{enumerate}[label = (\roman*)]
\item $0 \leq V^2_{1(f)} \leq 1$. 
\item $V^2_{1(f)} = 0$ if and only if a structure of null association exists in the table $($i.e., $\{p_{ij} = p_{i \cdot}p_{\cdot j}\})$.
\item $V^2_{1(f)} = 1$ if and only if a structure of complete association exists. For each column $j$ $(j = 1, 2, \dots , c)$, $i_j$ uniquely exists such that $p_{i_j, j} > 0$ and $p_{ij} = 0$ for all other $i(\neq i_j)$ $($assuming $p_{i\cdot} > 0$ for all i$)$.
\end{enumerate}
\end{theorem}

The proof of Theorem \ref{thm1} is provided in the Appendix.
Similar to the interpretation of measure $V^2_{1(\lambda)}$, $V^2_{1(f)}$ indicates the degree to which the prediction of the row category of an individual may be improved if knowledge regarding the column category of the individual exists. 
In this sense, $V^2_{1(f)}$ shows the strength of association between the row and column variables. 
The examples of the $f$-divergence are given below. 
When $f(x) = x \log x$, $I_f(\{p_{ij}\};\{p_{i \cdot}p_{\cdot j}\})$ is identical to the KL-divergence and $V^2_{1(f)}$ represented by
\begin{align*}
V_{KL} &= \frac{\displaystyle \sum^r_{i=1} \sum^c_{j=1}p_{ij}\log \left(\frac{p_{ij}}{p_{i \cdot}p_{\cdot j}}  \right)} {- \displaystyle \sum^r_{i=1} p_{i\cdot}\log p_{i\cdot}} 
\end{align*}
and $V_{KL}$ is identical to the Thile's uncertainty coefficient $U$ (see, \citealp{theil1970estimation}).
When $f(x) = x^2-x$, the Pearson's divergence is derived, and $V^2_{1(f)}$ is identical to Cram\'er's coefficient $V^2$ with $r \leq c$, and $f(x)=(x^{\lambda+1}-x)/\lambda(\lambda+1)$, $V^2_{1(f)}$ is identical to the power-divergence-type measure 
\begin{align*}
V^2_{1(\lambda)} &= \frac{\displaystyle \sum^r_{i=1} \sum^c_{j=1}p_{ij} \left[ \left(\frac{p_{ij}}{p_{i \cdot}p_{\cdot j}}  \right)^{\lambda} - 1 \right]} {\displaystyle \sum^r_{i=1} p^{1-\lambda}_{i\cdot} - 1}. 
\end{align*}
Further, in the case of $f(x) = (x-1)^2/(\theta x + 1 - \theta) + (x-1)/(1 - \theta)$ for $0 \leq \theta < 1$, $I_f(\{p_{ij}\};\{p_{i \cdot}p_{\cdot j}\})$ is identical to the $\theta$-divergence and $V^2_{1(f)}$ represented by the $\theta$-divergence-type measure
\begin{align*}
V^2_{1(\theta)}  &= \frac{\displaystyle \sum^r_{i=1} \sum^c_{j=1}\frac{(p_{ij}-p_{i \cdot}p_{\cdot j})^2 }{\theta p_{ij} + (1 - \theta)p_{i \cdot}p_{\cdot j}}} {\displaystyle \sum^r_{i=1} \frac{p_{i\cdot}(1-p_{i\cdot})}{(1-\theta) \left( \theta + (1-\theta)p_{i\cdot} \right) } }. 
\end{align*}
Measure $V^2_{1(\theta)}$, such as $V^2_{1(\lambda)}$, is also a single measure and one of the generalizations of $V^2$, which agrees with $V^2$ at $\theta = 0$. The numerator coincides with the triangular discrimination $\Delta$ at $\theta = 0.5$ (see, \citealp{dragomir2000new, topsoe2000some}).
Unlike the power-divergence, the $\theta$-divergence can measure departures from independence similar to the Euclidean distance, especially in the case of the triangular discrimination $\Delta$, which can measure symmetrical distances of $\{p_{ij}\}$and $\{p_{i \cdot}p_{\cdot j}\}$.
In the numerical experiments discussed in Sections $4$ and $6$, we treat the $\theta$-divergence-type measure as the example of a new single-parameter measure that can be considered by extending $V^2$ and compare it with the conventional one.
Moreover, analysis corresponding to various contingency tables can be performed by changing the function.

\subsection{Case II}
For the asymmetric situation wherein the row and column variables are the explanatory and response variables, respectively, we propose the following measure, which presents the strength of association between the row and column variables:
\begin{align*}
V^2_{2(f)} &= \frac{I_f(\{p_{ij}\};\{p_{i \cdot}p_{\cdot j}\})}{K_{2(f)}},
\end{align*}
where
\begin{align*}
K_{2(f)} &= \sum^c_{j=1}p^2_{\cdot j} f\left( \frac{1}{p_{\cdot j}} \right).
\end{align*}
Therefore, the following theorem is obtained for measure $V^2_{2(f)}$.
\begin{theorem}\label{thm2}
For each convex function $f$, 
\begin{enumerate}[label = (\roman*)]
\item $0 \leq V^2_{2(f)} \leq 1$. 
\item $V^2_{2(f)} = 0$ if and only if a structure of null association exists in the table $($i.e., $\{p_{ij} = p_{i \cdot}p_{\cdot j}\})$.
\item $V^2_{2(f)} = 1$ if and only if a structure of complete association exists; that is, for each row $i$ $(i = 1, 2, \ldots, r)$, $j_i$ uniquely exists such that $p_{i, j_i} > 0$ and $p_{ij} = 0$ for all other $j(\neq j_i)$ $($ assuming $p_{\cdot j} > 0$ for all j$)$.
\end{enumerate}
\end{theorem}
The proof of Theorem \ref{thm2} is obtained in a similar manner to the proof of Theorem \ref{thm1}.
$V^2_{2(f)}$ coincides with the value of $V^2_{1(f)}$ when the row and column variables are interchanged in the table, and $V^2_{2(f)}$ has no special characteristics compared to $V^2_{1(f)}$.
However, it is proposed because of its importance in Case III.

\subsection{Case III}
In an $r \times c$ contingency table wherein explanatory and response variables are undefined, using $V^2_{1(f)}$ and $V^2_{2(f)}$ is inappropriate if we are interested in determining the degree to what knowledge about the value of one variable can help us predict the value of the other variable. For this asymmetric situation, we propose the following measure that combines the ideas of both $V^2_{1(f)}$ and $V^2_{2(f)}$:
\begin{align*}
V^2_{3(f)} &= h^{-1} \left( \left(  w_1h\left( V^2_{1(f)} \right) + w_2h\left(V^2_{2(f)} \right)  \right) \right),
\end{align*}
where $h$ is the monotonic function and $w_1 + w_2 = 1$ $(w_1, w_2 \geq 0)$. Then, the following theorem is attained for measure $V^2_{3(f)}$.
\begin{theorem}\label{thm3}
For each convex function $f$, 
\begin{enumerate}[label = (\roman*)]
\item $0 \leq V^2_{3(f)} \leq 1$. 
\item $V^2_{3(f)} = 0$ if and only if a structure of null association exists in the table $($i.e., $\{p_{ij} = p_{i \cdot}p_{\cdot j}\})$.
\item $V^2_{3(f)} = 1$ if and only if a structure of complete association exists; that is, at most one non zero probability appears in each row or each column (assuming all marginal probabilities are non zero).
\end{enumerate}
\end{theorem}

The proof of Theorem \ref{thm3} is provided in the Appendix.
We can show that, if $h(u) = \log u$ and $w_1=w_2$, $V^2_{3(f)}$ is denoted by 
\begin{align*}
V^2_{G(f)} &= \frac{I_f(\{p_{ij}\};\{p_{I \cdot}p_{\cdot j}\})}{\sqrt{K_{1(f)}K_{2(f)}}} = \sqrt{V^2_{1(f)} V^2_{2(f)}},
\end{align*}
and if $h(u) = 1/u$ and $w_1=w_2$, $V^2_{3(f)}$ is represented by 
\begin{align*}
V^2_{H(f)} &= \frac{2I_f(\{p_{ij}\};\{p_{i \cdot}p_{\cdot j}\})}{K_{1(f)} + K_{2(f)}} = \frac{2V^2_{1(f)} V^2_{2(f)}}{V^2_{1(f)} + V^2_{2(f)}}.
\end{align*}
Notably, $V^2_{G(f)}$ and $V^2_{H(f)}$ are the geometric mean and harmonic mean of $V^2_{1(f)}$ and $V^2_{2(f)}$, respectively. We confirm that, when $f(x) = x^2-x$ with $r=c$, $V^2_{3(f)}$ is identical to Cram\'er's coefficient $V^2$. Conversely, for $f(x)=(x^{\lambda+1}-x)/\lambda(\lambda+1)$, $V^2_{3(f)}$ is consistent with Miyamoto's measure $G^2_{(\lambda)}$ (\citealp{miyamoto2007generalized}). 

For an $r \times r$ contingency table with the same row and column classifications, $V^2_{3(f)}=1$ if and only if the main diagonal cell probabilities in the $r \times r$ table are nonzero and the off-diagonal cell probabilities are all zero after interchanging some row and column categories. 
Therefore, all observations concentrate on the main diagonal cells.
While predicting the values of categories of an individual, $V^2_{3(f)}$ would specify the degree to which the prediction could be improved if knowledge about the value of one variable exists. In this sense, $V^2_{3(f)}$ also indicates the strength of association between the row and column variables. 
If only the marginal distributions $\{p_{i\cdot}\}$ and $\{p_{\cdot j}\}$ are known, we consider predicting the values of the individual row and column categories in terms of probabilities with independent structures.
\begin{theorem}\label{thm4}
For any fixed convex functions $f$ and monotonic functions $h$,
\begin{enumerate}
\item $\min(V^2_{1(f)}, V^2_{2(f)}) \leq V^2_{3(f)} \leq \max(V^2_{1(f)}, V^2_{2(f)})$,  
\item $\min(V^2_{1(f)}, V^2_{2(f)}) \leq V^2_{H(f)} \leq V^2_{G(f)} \leq \max(V^2_{1(f)}, V^2_{2(f)})$. 
\end{enumerate}
\end{theorem}
The proof of Theorem \ref{thm4} is provided in the Appendix.
When  $f(x) = x^2-x$ with $r=c$, we observe that $V^2_{1(f)} = V^2_{2(f)} = V^2_{3(f)} = V^2_{H(f)} = V^2_{G(f)} = V^2$ (being the Cram\'er’s coefficient).

\section{Relationship between measures and bivariate normal distribution}\label{sec3}
In the analysis of the two-way contingency table, \cite{tallis1962maximum}, \cite{lancaster1964estimation}, \cite{kirk1973numerical}, and \cite{divgi1979calculation} proposed an approach based on the bivariate normal distribution.
This approach assumes that the classification of rows and columns results from continuous random variables with a bivariate normal distribution, that is, the sample contingency table comes from a discretized bivariate normal distribution.
In many contexts, this assumption is invalid and a more general approach is needed.
Therefore, \cite{goodman1981association, goodman1985analysis} presented an approximation close to the correlation structure of discrete bivariate distributions based on the association model, and \cite{becker1989bivariate} also made a similar proposal based on the KL-divergence.
Assuming a bivariate normal distribution is important for examining the correlation structure of the contingency table, and previous studies have considered the association based on the model.
In this section, we explain the relationship between the measures $V^2_{t(f)}$ ($t = 1,2,3$) and the correlation coefficient $\rho$ when a bivariate normal distribution can be assumed for the latent variables in the contingency table.

Assuming a latent variable, the ($i$,$j$) cell probability $p_{ij}$ of the $r \times c$ contingency table is denoted as
\begin{align*}
p_{ij} &= P(X = i, Y=j) \\
&= P(x_{i-1} < X^* \leq x_i, y_{j-1} < Y^* \leq y_j) \\
&= f_{X^*, Y^*}(\tilde{x}_i, \tilde{y}_j)\Delta_{x_i} \Delta_{y_j},
\end{align*}
where $x_{i-1} < \tilde{x}_i \leq x_i$, $y_{j-1} < \tilde{y}_j \leq y_j$ and  $f_{X^*, Y^*}(\tilde{x}_i, \tilde{y}_j)$ is a continuous joint density function of random variables $X^*$ and $Y^*$.
$\Delta_{x_i}$ and $\Delta_{y_j}$ are the width of intervals $(x_{i-1}, x_{i}]$ and $(y_{j-1}, y_{j}]$, respectively. 
In this situation, it is possible to approximate $I_f(\{p_{ij}\};\{p_{i \cdot}p_{\cdot j}\})$ as follows:
\begin{align}
\begin{split}\label{approximateV2}
&I_f(\{p_{ij}\};\{p_{i \cdot}p_{\cdot j}\}) \\
&= \sum^r_{i=1} \sum^c_{j=1} f_{X^*}(\tilde{x}_i) f_{Y^*}(\tilde{y}_j) f\left(\frac{f_{X^*,Y^*}(\tilde{x}_i, \tilde{y}_j)}{f_{X^*}(\tilde{x}_i) f_{Y^*}(\tilde{y}_j)} \right)\Delta_{x_i} \Delta_{y_j} \\
&\xrightarrow[\Delta_{x_i} \Delta_{y_j} \to 0]{} \int^{\infty}_{-\infty} \int^{\infty}_{-\infty} f_{X^*}(x) f_{Y^*}(y) f\left(\frac{f_{X^*,Y^*}(x, y)}{f_{X^*}(x) f_{Y^*}(y)} \right) dx dy,
\end{split}
\end{align}
where $f_{X^*}(x)$ and $f_{Y^*}(y)$ are marginal probability density functions of $f_{X^*,Y^*}(x, y)$.

Let $X^*$ and $Y^*$ be random variables according to the bivariate normal distribution and joint density function is
\begin{align*}
\begin{split}
f_{X^*,Y^*}(x, y) &= \frac{1}{2\pi \sigma_x \sigma_y \sqrt{1-\rho^2}} \exp \left[ -\frac{1}{2(1-\rho^2)} \right. \\
&\quad \left. \qquad \left\lbrace \left(\frac{x-\mu_x}{\sigma_x} \right)^2 - 2\rho \left(\frac{x-\mu_x}{\sigma_x}  \right) \left(\frac{y-\mu_y}{\sigma_y}  \right) + \left(\frac{y-\mu_y}{\sigma_y}  \right)^2 \right\rbrace \right] \\
&\quad -\infty < x < +\infty, \quad -\infty < y < +\infty 
\end{split} 
\end{align*}
where $\rho$ is the correlation coefficient between $X^*$ and $Y^*$. 
The value of the correlation coefficient ranges from $-1$ to $1$. 
In the formula, the standard deviation $\sigma_x$ and $\sigma_y$ are positive constants. However, the means $\mu_x$ and $\mu_y$ do not have to be positive constants.
When applying $f(x)=(x^{\lambda+1}-x)/\lambda(\lambda+1)$, the relationship between the power-divergence and correlation coefficient $\rho$ is expressed as
\begin{align}
I_f(\{p_{ij}\};\{p_{i \cdot}p_{\cdot j}\}) &\approx \frac{1}{\lambda(\lambda + 1)} \left\lbrace (1-\rho^2)^{-\frac{\lambda}{2}}(1-\lambda^2 \rho^2)^{-\frac{1}{2}}-1 \right\rbrace,
\label{V2vslambda1}
\end{align}
where $\lambda < 1/\vert \rho\vert$.
Therefore, it is better to use less than 1 for $\lambda$ under the assumption.
If we want to capture the relationship between the measures and correlation coefficient $\rho$, by applying the value at $\lambda = 0$, which is assumed to be the continuous limit as $\lambda \rightarrow 0$ (i.e $f(x) = x\log x$), it can be expressed as 
\begin{align}
I_f(\{p_{ij}\};\{p_{i \cdot}p_{\cdot j}\}) &\approx - \frac{1}{2}\log(1-\rho^2).
\label{V2vslambda2}
\end{align}
When we consider the latent variable and approximate a divergence, the relationship can be shown as (\ref{V2vslambda1}) and (\ref{V2vslambda2}).
These equations show that the value is monotonically increasing with respect to $\vert \rho\vert$.
Therefore, by considering the measures, the relationship can be captured and an upper limit can be established.

This section showed the relationship between the measures and correlation coefficient $\rho$ using the bivariate normal distribution and $f(x)=(x^{\lambda+1}-x)/\lambda(\lambda+1)$ as examples.
However, in the $\theta$-divergence and more general divergence cases, it is difficult to calculate (\ref{approximateV2}) in a closed form.
Therefore, in the next section, we confirm that the value of the measures increases monotonically as the correlation coefficient moves away from $0$, even when the $\theta$-divergence is applied.

\section{Numerical study}\label{sec4}
This section compares the measurements by function or parameter.
In the numerical study, we use artificial data generated from discrete bivariate distributions with zero means and unit variances, as in \cite{goodman1981association, goodman1985analysis}, and \cite{becker1989bivariate}.
The method of partitioning the bivariate normal distribution is to use cut-points that generate uniform marginal distributions.
For instance, when creating a $4\times 4$ probability table, we split the bivariate normal distribution using $z_{0.25}$, $z_{0.50}$, and $z_{0.75}$ as cut-points.
The $4\times 4$ artificial probability tables created for the numerical study are given in the Appendix.
The benefit of this method is that the strength of association between the row and column variables in the contingency table is known from the bivariate normal distribution, which is appropriate for examining the measures.
For the comparison of the measures, we use Tomizawa's power-divergence-type measures ($f(x)=(x^{\lambda+1}-x)/\lambda(\lambda+1)$ for $0 \leq \lambda \leq 1$) and the newly proposed $\theta$-divergence-type measures ($f(x) = (x-1)^2/(\theta x + 1 - \theta) + (x-1)/(1 - \theta)$ for $0 \leq \theta < 1$), both of which are a single-parameter divergence and extensions of Cram\'er's coefficient $V^2$.

Table \ref{rrxz} presents the values of the measures $V^2_{t(f)}$ ($t=1,2,3$) for each $4\times 4$ probability tables with $\rho = 0.0, 0.4, 0.8, 1.0$. 
Notably, in the case of $r \times r$ artificial contingency tables, each of $\{p_{i \cdot}\}$ and $\{p_{\cdot j}\}$ is constant so that $V^2_{1(f)} = V^2_{2(f)} = V^2_{3(f)}$.
Table \ref{rrxz} shows that, when the correlation is away from $0$, $\hat{V}^2_{t(f)}$ are close to $1.0$. 
Further, $\rho = 0$ if and only if the measures show that a structure of null association exists in the table, and $\rho = 1.0$ if and only if the measures confirm that a structure of complete association exists.
The sharp increase around $\rho = 1.0$ can be explained by the previous section's relationship between the measures and correlation coefficients $\rho$.
Another important finding is how each measure increases at $\rho = 0.4, 0.8$.
In the case of $V^2$ ($\lambda = 1, \theta = 0$), the increasing trend of $V^2$ with the change of $\rho$ is slower than most measures.
It may not be possible to accurately determine small differences in the strength of association when comparing multiple contingency tables, so having a broad perspective by extension may allow careful analysis.
These results suggest that $V^2$ may not accurately determine small differences in the strength of association when comparing multiple contingency tables made by the bivariate normal distribution.
The same is true for the power-divergence-type measures, which have an increasing trend similar to $V^2$.
We may consider that it is better to use the $\theta$-divergence-type measures with $\theta=0.7$ in order to determine the small differences in the strength of association.
Values of $V^2_{t(f)}$ for other $\rho$, and coverage probabilities are provided in the Appendix.
\begin{table}[hbtp]
\caption{Values of measures $V^2_{t(f)}$ $(t =1, 2, 3)$ setting (a) the power-divergence for any $\lambda$ and (b) the $\theta$-divergence for any $\theta$ in $4\times 4$ probability tables with $\rho = 0, 0.4, 0.8, 1.0$.}
\label{rrxz}
\centering
\begin{multicols}{2}
\begin{tabular}{ccccc}
\multicolumn{5}{c}{(a) applying the power-divergence}  \\
\hline
& \multicolumn{4}{c}{the correlation $\rho$} \\ \cline{2-5}
$\lambda$ & 0.0 & 0.4  & 0.8 & 1.0 \\ \hline
0.0 & 0.000 & 0.046 & 0.254 & 1.000 \\
0.2 & 0.000 & 0.048 & 0.255 & 1.000 \\
0.4 & 0.000 & 0.048 & 0.251 & 1.000 \\
0.6 & 0.000 & 0.046 & 0.244 & 1.000 \\
0.8 & 0.000 & 0.044 & 0.234 & 1.000 \\
1.0 & 0.000 & 0.042 & 0.224 & 1.000 \\  \hline
\end{tabular}

\begin{tabular}{ccccc}
\multicolumn{5}{c}{(b) applying the $\theta$-divergence }  \\
\hline
& \multicolumn{4}{c}{the correlation $\rho$}  \\ \cline{2-5}
$\theta$ & $0.0$ & $0.4$& $0.8$ & $1.0$ \\ \hline
0.0 & 0.000 & 0.042 & 0.224 & 1.000 \\
0.1 & 0.000 & 0.049 & 0.249 & 1.000 \\
0.3 & 0.000 & 0.055 & 0.278 & 1.000 \\
0.5 & 0.000 & 0.053 & 0.285 & 1.000 \\
0.7 & 0.000 & 0.042 & 0.266 & 1.000 \\
0.9 & 0.000 & 0.019 & 0.191 & 1.000 \\   \hline
\end{tabular}
\end{multicols}
\end{table}

\section{Approximate confidence intervals for measure}\label{sec5}
In the previous section, we confirmed the values of the proposed measures with simulated data.
However, when analyzing real data, $p_{ij}$ is unknown, and these values are also unknown. Hence, it is necessary to construct confidence intervals.
Therefore, in this section, we construct asymptotic confidence intervals by using the delta method.
$\{n_{ij}\}$ denotes the observed frequency from multinomial distribution, and $n$ denotes the total number of observations, namely, $\sum^r_{i=1} \sum^c_{j=1}n_{ij}$. 
The approximate standard error and large-sample confidence interval are obtained for $V^2_{t(f)} (t =1, 2, 3)$ using the delta method, which is described in, for example, \cite{agresti2003categorical, bishop2007discrete}, etc. 
The estimator of $V^2_{t(f)}$ (i.e., $\hat{V}^2_{t(f)}$) is given by $V^2_{t(f)}$ with $\{p_{ij}\}$ replaced by $\{ \hat{p}_{ij}\}$, where $\hat{p}_{ij} = n_{ij}/n$. 
When using the delta method, $\sqrt{n}(\hat{V}^2_{t(f)} - V^2_{t(f)})$ has a asymptotically normal distribution (i.e., as $n \to \infty$) with mean $0$ and variance $\sigma^2[V^2_{t(f)}]$. 
Refer to the Appendix for the values of $\sigma^2[V^2_{t(f)}]$.

We define $f(x)$ as once-differentiable and strictly convex and $f'(x)$ as a first derivative of $f(x)$ with respect to $x$. 
Assume $\hat{\sigma}^2[V^2_{t(f)}]$ be $\sigma^2[V^2_{t(f)}]$ with $\{p_{ij}\}$ replaced by $\{ \hat{p}_{ij}\}$. 
Then, an estimated standard error of $\hat{V}^2_{t(f)}$ is $\hat{\sigma}[V^2_{t(f)}] / \sqrt{n}$, and an approximate $100(1-\alpha)$ percent confidence interval of $\hat{V}^2_{t(f)}$ is $\hat{V}^2_{t(f)} \pm z_{\alpha/2} \hat{\sigma}[V^2_{t(f)}] / \sqrt{n}$, where $z_{\alpha/2}$ is the upper $\alpha/2$ percentage point from the standard normal distribution.

\section{Examples}\label{sec6}
In this section, we explain the benefits of using the f-divergence to extend the measures, with some actual data examples. 
We use Tomizawa's power-divergence-type measures ($f(x)=(x^{\lambda+1}-x)/\lambda(\lambda+1)$ for $\lambda = 0.0, 0.6, 1.0, 1.2, 1.5$) and the newly proposed $\theta$-divergence-type measures ($f(x) = (x-1)^2/(\theta x + 1 - \theta) + (x-1)/(1 - \theta)$ for $\theta = 0.0, 0.3, 0.5, 0.7, 0.9$), both of which are a single-parameter divergence and extensions of Cram\'er's coefficient $V^2$.
Let observe the estimates of the measures and confidence intervals.

\subsection*{Example 1}
Consider the data in Table \ref{race}, taken from the 2006 General Social Survey.
These are data, which show the relationship between family income and education in the United States separately for black and white categories of race.
By applying the measures $V^2_{1(f)}$, we consider to what educational degree can be improved when the prediction of family income for black and white categories of an individual is known.
\begin{table}[hbtp]
\caption{Data on educational degrees and family income, by race}
\centering
\label{race}
\begin{tabular}{lcccc}
\multicolumn{5}{c}{(a) Black Person} \\
\hline
 & \multicolumn{3}{c}{Family Income} & \\ \cline{2-4}
Highest Degree & Blow Average & Average  & Above Average & Total  \\ \hline
$<$ High school & 43 & 36 & 5 & 84  \\
High school, junior college & 104 & 140 & 23 & 267 \\
College, grad school & 16 & 30 & 18 & 64 \\ \hline
Total & 163 & 206 & 46  & 415 \\ \hline
\multicolumn{2}{l}{Source: 2006 General Social Survey}  \\
\end{tabular}
\\
\begin{tabular}{lcccc}
\multicolumn{5}{c}{(b) White Person} \\
\hline
 & \multicolumn{3}{c}{Family Income} & \\ \cline{2-4}
Highest Degree & Blow Average & Average  & Above Average & Total  \\ \hline
$<$ High school & 114 & 97 & 12 & 223  \\
High school, junior college & 410 & 658 & 221 & 1289 \\
College, grad school & 97 & 259 & 287 & 643 \\ \hline
Total & 621 & 1014 & 520  & 2155 \\ \hline
\multicolumn{2}{l}{Source: 2006 General Social Survey}  \\
\end{tabular}
\end{table}

Table \ref{race_m} shows the estimates of the measures, standard errors, and $95\%$ confidence intervals. Tables \ref{race_m}(a1, b1) and \ref{race_m}(b1, b2) show the results of the analysis of Tables \ref{race}(a) and \ref{race}(b), respectively.
One interesting finding is the confidence intervals for all $V^2_{1(f)}$ do not contain zero for any $\lambda$ and $\theta$. 
The results show that the two actual data have an associated structure from a point of view, other than Cram\'er's coefficient $V^2$.
Another important finding is the comparison of the confidence intervals.
For conventional the power-divergence-type measures, a comparison of Tables \ref{race_m}(a1) and \ref{race_m}(b1) shows that confidence intervals overlap for each $\lambda$.
Meanwhile, when $\theta = 0.9$ in Tables \ref{race_m}(a2) and \ref{race_m}(b2), the confidence intervals do not overlap. Table \ref{race}(a), where the estimate is closer to $0$, has higher independence.
Therefore, this analysis revealed the merit of using the measures extended with the $f$-divergence to express the differences that did not appear in the conventional one.
\begin{table}[hbtp]
\caption{Estimate of the measure $V^2_{1(f)}$, estimated approximate standard error for $\hat{V}^2_{1(f)}$, and approximate $95\%$ confidence interval of $V^2_{1(f)}$ applying (a1, b1) the power-divergence for any $\lambda$ and (a2, b2) the $\theta$-divergence for any $\theta$.}
\label{race_m}
\centering
\begin{multicols}{2}
\begin{tabular}{cccc}
\multicolumn{4}{c}{(a1) applying the power-divergence}   \\
\hline
$\lambda$ & $\hat{V}^2_{1(f)}$ & SE & $95\%$CI \\  \hline
0.0 & 0.032 & 0.014 & (0.005, 0.059) \\
0.6 & 0.036 & 0.016 & (0.005, 0.067) \\
1.0 & 0.034 & 0.016 & (0.003, 0.064) \\
1.2 & 0.032 & 0.015 & (0.002, 0.061) \\
1.5 & 0.028 & 0.014 & (0.001, 0.056) \\ \hline
\end{tabular}

\begin{tabular}{cccc}
\\
\multicolumn{4}{c}{(a2) applying the $\theta$-divergence}  \\
\hline
$\theta$ & $\hat{V}^2_{1(f)}$ & SE & $95\%$CI  \\  \hline
0.0 & 0.034 & 0.016 & (0.003, 0.064) \\
0.3 & 0.040 & 0.017 & (0.007, 0.072) \\
0.5 & 0.034 & 0.014 & (0.006, 0.062) \\
0.7 & 0.024 & 0.010 & (0.004, 0.044) \\
0.9 & 0.009 & 0.004 & (0.001, 0.017) \\ \hline
\end{tabular}

\begin{tabular}{cccc}
\multicolumn{4}{c}{(b1) applying the power-divergence}   \\
\hline
$\lambda$ & $\hat{V}^2_{1(f)}$ & SE & $95\%$CI \\  \hline
0.0 & 0.068 & 0.008 & (0.052, 0.083) \\
0.6 & 0.070 & 0.008 & (0.054, 0.086) \\
1.0 & 0.062 & 0.007 & (0.047, 0.076) \\
1.2 & 0.056 & 0.007 & (0.042, 0.069) \\
1.5 & 0.046 & 0.006 & (0.034, 0.057) \\ \hline
\end{tabular}

\begin{tabular}{cccc}
\\
\multicolumn{4}{c}{(b2) applying the $\theta$-divergence}  \\
\hline
$\theta$ & $\hat{V}^2_{1(f)}$ & SE & $95\%$CI  \\  \hline
0.0 & 0.062 & 0.007 & (0.047, 0.076) \\
0.3 & 0.084 & 0.010 & (0.065, 0.102) \\
0.5 & 0.076 & 0.009 & (0.059, 0.094) \\
0.7 & 0.058 & 0.007 & (0.044, 0.072) \\
0.9 & 0.026 & 0.004 & (0.018, 0.033) \\ \hline
\end{tabular}
\end{multicols}

\end{table}

\subsection*{Example 2}
Consider the data in Table \ref{eyes}, obtained from \cite{tomizawa1985analysis}. 
These tables provide information on the unaided distance vision of 4746 university students aged 18 to about 25 and 3168 elementary students aged 6 to about 12. 
In Table \ref{eyes}, the row and column variables are the right and left eye grades, respectively, with the categories ordered from the highest grade (1) to the lowest grade (4).
As the right and left eye grades have similar classifications, we apply measure $V^2_{H(f)}$.
\begin{table}[hbtp]
\caption{Unaided distance vision data for university and elementary students}
\centering
\label{eyes}
\begin{tabular}{lccccc}
\multicolumn{6}{c}{(a) University Students} \\
\hline
 & \multicolumn{4}{c}{Left eye grade} & \\ \cline{2-5}
Right eye grade & (1) & (2) & (3) & (4) & Total  \\ \hline
Highest (1) & 1291 & 130 & 40 & 22 & 1483  \\
Second (2) & 149 & 221 & 114 & 23 & 507 \\
Third (3) & 64 & 124 & 660 & 185 & 1033 \\
Lowest (4)  & 20 & 25 & 249 & 1429 & 1723 \\   \hline
Total & 1524 & 500 & 1063 & 1659  & 4746 \\ \hline
\multicolumn{2}{l}{Source: \cite{tomizawa1985analysis}}  \\
\end{tabular}

\begin{tabular}{lccccc}
\\
\multicolumn{6}{c}{(b) Elementary Students}\\
\hline
 & \multicolumn{4}{c}{Left eye grade} & \\ \cline{2-5}
Right eye grade & (1) & (2) & (3) & (4) & Total  \\ \hline
Highest (1) & 2470 & 126 & 21 & 10 & 2627  \\
Second (2) & 96 & 138 & 33 & 5 & 272 \\
Third (3) & 10 & 42 & 75 & 15 & 142 \\
Lowest (4)  & 12 & 7 & 16 & 92 & 127 \\   \hline
Total & 2588 & 313 & 145 & 122  & 3168 \\ \hline
\multicolumn{2}{l}{Source: \cite{tomizawa1985analysis}} 
\end{tabular}
\end{table}

Table \ref{eyesr} provides the estimates of the measures, standard errors, and confidence intervals.
Tables \ref{eyesr}(a1, b1) and \ref{eyesr}(b1, b2) show the results of the analysis of Tables \ref{eyes}(a) and \ref{eyes}(b), respectively.
The results of this analysis show that the two actual data have a strong structure of association in terms of the estimates and confidence intervals for all $\lambda$ and $\theta$.
After comparing the value of the measures between Tables \ref{eyesr}(a1, b1) and \ref{eyesr}(b1, b2), we found that the strength of association between the right and left eyes is greater for elementary school students in terms of the estimates for each parameter.
After comparing the confidence intervals between Tables \ref{eyesr}(a1, b1) and \ref{eyesr}(b1, b2) can be similarly concluded, we found that the confidence intervals for each $\lambda$ overlap, but not when $\theta = 0.5, 0.7, 0.9$.
Another interesting finding is that, unlike Example $1$, the confidence intervals of the results of the analysis in Example $2$ do not overlap when $\theta = 0.5, 0.7 $.
In terms of the triangular discrimination $\Delta$ ($\theta = 0.5$), which is not observed in the $V^2$ or the power-divergence-type measures, it can be assumed that this result provides evidence that Table \ref{eyes}(b) has a stronger association structure.
Therefore, the extension by the $f$-divergence helps us perform the analysis safely.
\begin{table}[hbtp]
\caption{Estimate of measures $V^2_{H(f)}$, estimated approximate standard error for $\hat{V}^2_{H(f)}$, and approximate $95\%$ confidence interval of $V^2_{H(f)}$ applying (a1, b1) the power-divergence for any $\lambda$ and (a2, b2) the $\theta$-divergence for any $\theta$.}
\label{eyesr}
\centering
\begin{multicols}{2}
\begin{tabular}{cccc}
\multicolumn{4}{c}{(a1) applying the power-divergence}   \\
\hline
$\lambda$ & $\hat{V}^2_{H(f)}$ & SE & $95\%$CI \\  \hline
0.0 & 0.413 & 0.017 & (0.379, 0.446) \\
0.6 & 0.400 & 0.018 & (0.364, 0.436) \\
1.0 & 0.357 & 0.020 & (0.317, 0.397) \\
1.2 & 0.334 & 0.022 & (0.291, 0.376) \\
1.5 & 0.300 & 0.023 & (0.255, 0.346) \\ \hline
\end{tabular}

\begin{tabular}{cccc}
\\
\multicolumn{4}{c}{(a2) applying the $\theta$-divergence}  \\
\hline
$\theta$ & $\hat{V}^2_{H(f)}$ & SE & $95\%$CI \\  \hline
0.0 & 0.357 & 0.020 & (0.317, 0.397) \\
0.3 & 0.479 & 0.018 & (0.444, 0.514) \\
0.5 & 0.453 & 0.019 & (0.416, 0.490) \\
0.7 & 0.384 & 0.020 & (0.345, 0.423) \\
0.9 & 0.223 & 0.019 & (0.185, 0.261) \\ \hline
\end{tabular}

\begin{tabular}{cccc}
\multicolumn{4}{c}{(b1) applying the power-divergence}   \\
\hline
$\lambda$ & $\hat{V}^2_{H(f)}$ & SE & $95\%$CI \\  \hline
0.0 & 0.461 & 0.009 & (0.443, 0.479) \\
0.6 & 0.445 & 0.009 & (0.428, 0.463) \\
1.0 & 0.402 & 0.009 & (0.384, 0.420) \\
1.2 & 0.375 & 0.009 & (0.357, 0.392) \\
1.5 & 0.330 & 0.009 & (0.312, 0.348) \\ \hline
\end{tabular}

\begin{tabular}{cccc}
\\
\multicolumn{4}{c}{(2b) applying the $\theta$-divergence}  \\
\hline
$\theta$ & $\hat{V}^2_{H(f)}$ & SE & $95\%$CI \\  \hline
0.0 & 0.402 & 0.009 & (0.384, 0.420) \\
0.3 & 0.509 & 0.009 & (0.492, 0.527) \\
0.5 & 0.510 & 0.009 & (0.492, 0.528) \\
0.7 & 0.472 & 0.010 & (0.452, 0.492) \\
0.9 & 0.337 & 0.013 & (0.311, 0.363) \\ \hline
\end{tabular}
\end{multicols}

\end{table}

\begin{remark}\label{rem1}
(Brief guideline for choosing functions and parameters)
Our proposed measures $V^2_{t(f)} (t=1,2,3)$ have functions $f(x)$ and the function's parameters, which are set by the user.
\cite{momozaki2022extension} discussed how these should be selected, and we also agree with the discussion.

We recommend using various values instead of only one.
$V^2_{t(f)}$ are a broad class of measures that include, as special cases ,the power-divergence-type measures ($f(x)=(x^{\lambda+1}-x)/\lambda(\lambda+1)$) and the $\theta$-divergence ($f(x) = (x-1)^2/(\theta x + 1 - \theta) + (x-1)/(1 - \theta)$), which encompass Cram\'er's coefficient $V^2$ and others.
It may be useful when we are interested in exploring new contingency table data with correlated row and column variables.
%
While it is important to investigate various functions and parameters to give safety to the results of the analysis, the user should not select values for own convenience.
In the actual analysis, it is considered that the results can be evaluated as the safest by varying the tuning parameters with a function that includes the most reasonable distance in the user's research field.
However, when the reasonable distance is unclear, it is necessary to investigate with many functions because it is necessary to give considerations from various viewpoints (e.g., Squared $L_2$ family distance, Shannon's entropy family distance, etc).
In addition, some choice suggestions can be made by considering Cram\'er's coefficient $V^2$  itself.

Cram\'er's coefficient $V^2$ is a popular measure for evaluating the degree of association between row and column variables in two-way contingency tables, but there are several limitations.
\cite{kvaalseth2018alternative} points out some limitations of Cram\'er's coefficient $V^2$.
The main objective of this study is to generalize Cram\'er's coefficient $V^2$, but it is also possible to provide two improvements.
The first is that meaningful values are difficult to interpret.
If the function applied to the $f$-divergence is integrable in (\ref{approximateV2}) as well as the power-divergence in (\ref{V2vslambda1}) and (\ref{V2vslambda2}), it can give an operationally meaningful interpretation to the value of the measure corresponding to the correlation coefficient $\rho$.
The second is that the degree of association may be overestimated when the observation frequency is small.
This limitation also has Tomizawa's power-divergence-type measures $V^2_{t(\lambda)} (t=1,2,3)$, which have not been improved by generalization with the tuning parameter $\lambda$.
In such cases, an evaluation can be given from a point of view similar to the Euclidean distance using the $\theta$-divergence (except $\theta=0$) or other K-divergence, etc.
As an example, consider the article data in Table \ref{article}.
The data indicate clearly a very near independence between row and column categories with the Euclidean distance $\vert p_{ij}-p_{i\cdot}p_{\cdot j}\vert$ being either $0$ or $0.01$ and $\sum^3_{i=1} \sum^3_{j=1} \vert p_{ij}-p_{i\cdot}p_{\cdot j}\vert = 0.04$.
However, Cram\'er's coefficient and Tomizawa's power-divergence-type measure both have large values, while the $\theta$-divergence-type measure is close to the Euclidean distance in Table \ref{article_result}.
There are other limitations, but it may be possible to improve the limitations of Cram\'er's coefficient $V^2$ while ensuring the properties of the measure by the function to be applied.
\begin{table}[hbtp]
\caption{Artificial data to show differences from Cram\'er's coefficient $V^2$ and Tomizawa's power-divergence-type measures $V^2_{t(\lambda)} (t=1,2,3)$.}
\centering
\label{article}
\begin{tabular}{ccccc}
\hline
& (1) & (2) & (3) & Total  \\ \hline
(1) & 0.50 & 0 & 0.20 & 0.70  \\
(2) & 0 & 0.01 & 0.01 & 0.02 \\
(3) & 0.20 & 0 & 0.08 & 0.28 \\ \hline
Total & 0.70 & 0.01 & 0.29 & 1 \\ \hline
\end{tabular}
\end{table}
\begin{table}[hbtp]
\caption{Values of measures $V^2_{H(f)}$ setting (a) the power-divergence for any $\lambda$ and (b) the $\theta$-divergence for Table \ref{article}.}
\begin{multicols}{2}
\centering
\label{article_result}
(a) the power-divergence-type measure
\begin{tabular}{cc} \hline
$\lambda$ & $V^2_{H(f)}$  \\ \hline
0.0 & 0.081  \\
0.6 & 0.159  \\
1.0 & 0.254  \\
1.2 & 0.298  \\
1.5 & 0.335  \\ \hline
\end{tabular}

(b) the $\theta$-divergence-type measure
\begin{tabular}{cc} \hline
$\theta$ & $V^2_{H(f)}$  \\ \hline
0.0 & 0.254  \\
0.3 & 0.065  \\
0.5 & 0.058  \\
0.7 & 0.056  \\
0.9 & 0.055  \\ \hline
\end{tabular}
\end{multicols}
\end{table}
\end{remark}

\section{Conclusion}\label{sec7}
We found that the strength of association between the row and column variables in two-way contingency tables can be safely analyzed by proposing measures $V^2_{t(f)} (t =1, 2, 3)$ that generalizes Cram\'er's coefficient $V^2$ via the $f$-divergence.
First, this study proved that a measure applying a function $f(x)$ that satisfies the condition of the $f$-divergence has desirable properties for measuring the strength of association in contingency tables.
Hence, we can easily construct a new measure using a divergence that has essential properties for the analyst.
Furthermore, we can give a further interpretation of the association between rows and columns in contingency tables, which could not be obtained with a conventional one.
Second, we showed the relationship between the proposed measures $V^2_{t(f)}$ and the bivariate normal distribution.
We found that the relationship between the power-divergence and correlation coefficient $\rho$ is approximately formulated and more succinct with $\lambda = 0$.

Measures $V^2_{t(f)}$ always range between $0$ and $1$, independent of the dimensions $r$ and $c$ and the sample size $n$. 
Thus, comparing the strength of association between the row and column variables in several tables is useful. 
This is crucial in checking the relative magnitude of the strength of association between the row and column variables to the degree of complete association. 
Specifically, $V^2_{1(f)}$ ($V^2_{2(f)}$) would be effective when the row and column variables are the response (explanatory) and explanatory (response) variables, respectively, while $V^2_{3(f)}$ would be useful when explanatory and response variables are not defined.
Furthermore, we first need to check if independence is established by using a test statistic, such as Pearson's chi-squared statistics, to analyze the strength of association between the row and column variables. 
Then, if it is determined that a structure of the association exists, the next step would be to measure the strength of the association by using $V^2_{t(f)}$. 
However, if the table is determined as independent, employed $V^2_{t(f)}$ may not be meaningful. 
Furthermore, $V^2_{t(f)}$ is invariant under any permutation of the categories.
Therefore, we can apply it to the data analysis on a nominal or ordinal scale.

We observe that (i) the estimate of the strength of association should be considered in terms of an approximate confidence interval for $V^2_{t(f)}$ rather than $\hat{V}^2_{t(f)}$ itself and (ii) the measure helps describe relative magnitudes (of the strength of association), rather than absolute magnitudes. 

\section*{Acknowledgments}
This work was supported by JSPS Grant-in-Aid for Scientific Research (C) Number JP20K03756.

\appendix

\setcounter{equation}{0}
\def\theequation{A.\arabic{equation}}
\def\thethm{A.\arabic{thm}}
\def\thelem{A.\arabic{lem}}

\section{Proofs}\label{secA1}

Before proving the theorem 1 for the measure $V^2_{1(f)}$, the following lemma for the proof is introduced.
\begin{lem}\label{lem-cov}
Let f be a strictly convex function on $[0, +\infty)$ and
\begin{align}
f(x) = \left\{
\begin{array}{cc}
xg(x) & (x > 0)\\
0 & (x = 0)
\end{array}
. \right. \label{eq11g}
\end{align}
Subsequently, $g$ is a strictly monotonically increasing function.
\end{lem}
\begin{proof}[Proof of Lemma \ref{lem-cov}]

Due to the strictly convex function $f$, for any $x_1, x_2 \in [0, +\infty)$ and for any $p \in (0, 1)$, it holds that:
\begin{align*}
f(px_1 + (1-p)x_2) < pf(x_1) + (1-p)f(x_2).
\end{align*}
When $x_2 = 0$ and $x_1 \ne 0$, we have
\begin{align*}
f(px_1) < pf(x_1) &\iff px_1g(px_1) < px_1g(x_1) \\
&\iff g(px_1) < g(x_1).  
\end{align*}
Thus, $g$ is a strictly monotonically increasing function.
\end{proof}

\begin{proof}[Proof of Theorem 1]
The $f$-divergence is first transformed as follows: 
\begin{align*}
I_f(\{p_{ij}\};\{p_{i \cdot}p_{\cdot j}\}) &= \sum^r_{i=1} \sum^c_{j=1}p_{ij} \left( \frac{1}{\frac{p_{ij}}{p_{i \cdot}p_{\cdot j}}} f\left(\frac{p_{ij}}{p_{i \cdot}p_{\cdot j}}  \right)\right) \\
&= \sum^r_{i=1} \sum^c_{j=1}p_{ij} g\left(\frac{p_{ij}}{p_{i \cdot}p_{\cdot j}}  \right), 
\end{align*}
where $g$ is given by \eqref{eq11g}. From Lemma \ref{lem-cov}, since $g$ is the strictly monotonically increasing function, it holds that
\begin{align*}
I_f(\{p_{ij}\};\{p_{i \cdot}p_{\cdot j}\}) \leq \sum^r_{i=1} \sum^c_{j=1}p_{ij} g\left(\frac{1}{p_{i \cdot}}  \right) = \sum^r_{i=1}p^2_{i \cdot} f\left( \frac{1}{p_{i \cdot}} \right).
\end{align*}
Furthermore, from Jensen's inequality, we have
\begin{align}
I_f(\{p_{ij}\};\{p_{i \cdot}p_{\cdot j}\}) &\geq f\left(\sum^r_{i=1} \sum^c_{j=1}p_{i \cdot}p_{\cdot j} \frac{p_{ij}}{p_{i \cdot}p_{\cdot j}} \right) = 0.\label{lowbound}
\end{align}
Hence, $0 \leq V^2_{1(f)} \leq 1$ is obtained.

Afterward, $V^2_{1(f)} = 0$ follows from the property of the $f$-divergence $f\left(1\right) = 0$ if a structure of null association is observable (i.e., $\{p_{ij} = p_{i \cdot}p_{\cdot j}\}$). 
When $V^2_{1(f)} = 0$, it holds that: 
\begin{align*}
I_f(\{p_{ij}\};\{p_{i \cdot}p_{\cdot j}\}) &= \sum^r_{i=1} \sum^c_{j=1}p_{i \cdot}p_{\cdot j} f\left(\frac{p_{ij}}{p_{i \cdot}p_{\cdot j}}  \right) = 0.
\end{align*}
From the equality of \eqref{lowbound} holds, we have $p_{ij} / p_{i \cdot}p_{\cdot j} = 1$ for all $i,j$. 
Thus we obtain $\{p_{ij} = p_{i \cdot}p_{\cdot j}\}$ from the properties of the $f$-divergence.

Finally, if there uniquely exists $i_j$ for each column such that $p_{i_j, j} > 0$ and $p_{ij} = 0$ for all other $i(\neq i_j)$, the measure $V^2_{1(f)}$ can be expressed as:
\begin{align*}
V^2_{1(f)} &= \frac{\sum^r_{i=1} \sum^c_{j=1}p_{i \cdot}p_{ij} f\left(\frac{1}{p_{i \cdot}} \right)}{\sum^r_{i=1}p^2_{i \cdot} f\left( \frac{1}{p_{i \cdot}} \right)} = 1.
\end{align*}
Contrariwise, when $V^2_{1(f)} = 1$, we have
\begin{align*}
0 &= I_f(\{p_{ij}\};\{p_{i \cdot}p_{\cdot j}\}) - K_{1(f)} \\
&= \sum^r_{i=1} \sum^c_{j=1}p_{ij} \left( \frac{1}{\frac{p_{ij}}{p_{i \cdot}p_{\cdot j}}} f\left(\frac{p_{ij}}{p_{i \cdot}p_{\cdot j}}  \right)\right) -  \sum^r_{i=1} \sum^c_{j=1}p_{ij} \left( \frac{1}{\frac{1}{p_{i \cdot}}} f\left(\frac{1}{p_{i \cdot}} \right)\right) \\
&=  \sum^r_{i=1} \sum^c_{j=1}p_{ij}  \left(g\left(\frac{p_{ij}}{p_{i \cdot}p_{\cdot j}} \right) - g\left(\frac{1}{p_{i \cdot}} \right) \right). 
\end{align*}
From Lemma \ref{lem-cov}, as $g$ is the strictly monotonically increasing function, the equality is satisfied if for each column, there is only one $i_j$ such that $p_{i_j, j} > 0$ and $p_{ij} = 0$ for the other all $i(\neq i_j)$.
\end{proof}

\begin{proof}[Proof of Theorem 3]
First, for the weighted average of $h(V^2_{1(f)})$ and $h(V^2_{2(f)})$,
\begin{align*}
\min(h(V^2_{1(f)}), h(V^2_{2(f)})) \leq w_1h\left( V^2_{1(f)} \right) + w_2h\left(V^2_{2(f)} \right) \leq \max(h(V^2_{1(f)}), h(V^2_{2(f)})).
\end{align*}
From this correlation, we can show that 
\begin{align*}
\min(V^2_{1(f)}, V^2_{2(f)}) \leq V^2_{3(f)} \leq \max(V^2_{1(f)}, V^2_{2(f)}).
\end{align*}
As $0 \leq V^2_{1(f)} \leq 1$ and $0 \leq V^2_{2(f)} \leq 1$ originate from Theorems 1 and 2, $0 \leq V^2_{3(f)} \leq 1$ holds.

Subsequently, if $\{p_{ij} = p_{i \cdot}p_{\cdot j}\}$ because of $V^2_{1(f)} = V^2_{2(f)} = 0$ emerging from Theorems 1 and 2, $V^2_{3(f)} = 0$ is obvious. Conversely, if $V^2_{3(f)} = 0$, then we have
\begin{align*}
h(0) = w_1h\left( V^2_{1(f)} \right) + w_2h\left(V^2_{2(f)} \right)
\end{align*}
Notably, $h$ is the monotonic function, so the equality is satisfied at $V^2_{1(f)} = V^2_{2(f)} = 0$.
Hence, we obtain $\{p_{ij} = p_{i \cdot}p_{\cdot j}\}$ via Theorems 1 and 2.

Besides, $V^2_{3(f)} = 1$ is obvious in the case of (I) and (II), as $V^2_{1(f)} = V^2_{2(f)} = 1$ according to Theorems 1 and 2. Conversely, if $V^2_{3(f)} = 1$, then we have
\begin{align*}
h(1) = w_1h\left( V^2_{1(f)} \right) + w_2h\left(V^2_{2(f)} \right). 
\end{align*}
As mentioned previously, $h$ is the monotonic function, so the equality is satisfied at $V^2_{1(f)} = V^2_{2(f)} = 1$.
Thus, $V^2_{1(f)} = V^2_{2(f)} = 1$ is satisfied under situations (I) and (II).
\end{proof}

\begin{proof}[Proof of Theorem 4]
The inequality 1 in Theorem 4 has already been validated in the proof of Theorem 3, so it has been omitted.
We show the inequality 2.
Let $h_1(u) = 1/u$ and $h_2(u) = \log u$, then it holds that
\begin{align*}
V^2_{H(f)} &= h_1^{-1} \left( \left(  w_1h_1\left( V^2_{1(f)} \right) + w_2h_1\left(V^2_{2(f)} \right)  \right) \right) \\ 
&= h_2^{-1} \circ h_2 \circ h_1^{-1} \left( \left(  w_1h_1\left( V^2_{1(f)} \right) + w_2h_1\left(V^2_{2(f)} \right)  \right) \right).
\end{align*}
As $h_2 \circ h_1^{-1}$ is a convex function, 
\begin{align*}
V^2_{H(f)} &\leq h_2^{-1} \left( \left(  w_1h_2 \circ h_1^{-1} \circ h_1\left( V^2_{1(f)} \right) + w_2h_2 \circ h_1^{-1} \circ h_1\left(V^2_{2(f)} \right)  \right) \right) \\ 
&\leq h_2^{-1} \left( \left(  w_1h_2\left( V^2_{1(f)} \right) + w_2h_2\left(V^2_{2(f)} \right)  \right) \right) \\
&= V^2_{G(f)}.
\end{align*}
\end{proof}

\section{Asymptotic variance for measures}\label{secA2}
Using the delta method, $\sqrt{n}(\hat{V}^2_{t(f)} - V^2_{t(f)})$ ($t = 1, 2, 3$) has an asymptotic variance
\begin{flalign*}
& \text{(i)\quad}\sigma^2[V^2_{1(f)}] = \frac{1}{(K_{1(f)})^2}\left[ \sum^r_{i=1} \sum^c_{j=1}(\Delta_{1ij})^2p_{ij} - \left( \sum^r_{s=1} \sum^c_{t=1}\Delta_{1st}p_{st}  \right)^2  \right], &
\end{flalign*}
where
\begin{equation*}
\left\{
  \begin{aligned}
  &\Delta_{1ij} = G_{ij} + H_{ij} - V^2_{1(f)} E_{1ij}, \hspace{240pt}\\
  &G_{ij} = \sum^r_{i=1}p_{i \cdot} f\left(\frac{p_{ij}}{p_{i \cdot}p_{\cdot j}}  \right) + \sum^c_{j=1}p_{\cdot j} f\left(\frac{p_{ij}}{p_{i \cdot}p_{\cdot j}}  \right), \\
  & H_{ij} = f^{\prime}\left(\frac{p_{ij}}{p_{i \cdot}p_{\cdot j}}  \right) - \sum^r_{i=1} \frac{p_{ij}}{p_{\cdot j}} f^{\prime} \left( \frac{p_{ij}}{p_{i \cdot}p_{\cdot j}}  \right) - \sum^c_{j=1} \frac{p_{ij}}{p_{i \cdot}} f^{\prime} \left( \frac{p_{ij}}{p_{i \cdot}p_{\cdot j}}  \right), &\\
  & E_{1ij} = 2p_{i \cdot} f\left( \frac{1}{p_{i \cdot}} \right) - f^{\prime}\left( \frac{1}{p_{i \cdot}} \right),
  \end{aligned}
\right.
\end{equation*}

\begin{flalign*}
& \text{(ii)\quad}\sigma^2[V^2_{2(f)}] = \frac{1}{(K_{2(f)})^2}\left[ \sum^r_{i=1} \sum^c_{j=1}(\Delta_{2ij})^2p_{ij} - \left( \sum^r_{s=1} \sum^c_{t=1}\Delta_{2st}p_{st}  \right)^2  \right], &
\end{flalign*}
where
\begin{equation*}
\left\{
  \begin{aligned}
  & \Delta_{2ij} = G_{ij} + H_{ij} - V^2_{2(f)} E_{2ij}, \hspace{240pt}&\\
  & E_{2ij} = 2p_{\cdot j} f\left( \frac{1}{p_{\cdot j}} \right) - f^{\prime}\left( \frac{1}{p_{\cdot j}} \right), &
  \end{aligned}
\right.
\end{equation*}

\begin{flalign*}
& \text{(iii)\quad}\sigma^2[V^2_{3(f)}] = \sum^r_{i=1} \sum^c_{j=1}\left( \frac{\partial V^2_{3(f)}}{\partial p_{ij}} \right)^2p_{ij} - \left( \sum^r_{s=1} \sum^c_{t=1}\left( \frac{\partial V^2_{3(f)}}{\partial p_{st}} \right)p_{st}  \right)^2 . &
\end{flalign*}

\section{Additional numerical studies}
Appendix C presents tables of the results of numerical studies on measures by function and parameter.
In numerical study 1, we compare measures by function and parameter from $4\times 4$ artificial data with different degrees of association.
Numerical study 2 further investigates the effect of the number of rows and columns.
The method of partitioning the bivariate normal distribution is to use cut-points that generate uniform marginal distributions.
For instance, when creating a $4\times 4$ probability table, we split the bivariate normal distribution using $z_{0.25}$, $z_{0.50}$, and $z_{0.75}$ as cut-points.

The measures used in the survey are the Tomizawa's power-divergence type measures ($f(x)=(x^{\lambda+1}-x)/\lambda(\lambda+1)$ for $0 \leq \lambda \leq 1$) and the newly proposed $\theta$-divergence type measures ($f(x) = (x-1)^2/(\theta x + 1 - \theta) + (x-1)/(1 - \theta)$ for $0 \leq \theta < 1$), both of which are a single-parameter divergence and extensions of the Cram\'er's coefficient $V^2$.
Also, use parameters $\lambda = 0.0, 0.2, 0.4, 0.6, 0.8, 1.0$ and $\theta = 0.0, 0.1, 0.3, 0.5, 0.7, 0.9$.

\subsection*{Numerical study 1}
Consider Table \ref{rrxy} to examine the relationship between the degree of association and the measures.
These tables are $4 \times4 $probability tables obtained by discretizing the bivariate normal distribution, and the correlation coefficient $\rho$ corresponds to the associated structure of the contingency table.
Therefore, by changing the correlation coefficient $\rho$ by $0.2$ from $-1$ to $1$, it is possible to capture changes in the measures depending on the degree of association.
\begin{table}[hbtp]
\caption{The $4 \times 4$ probability tables, formed by using three cutpoints for each variable at $z_{0.25}, z_{0.50}, z_{0.75}$ from a bivariate normal distribution with zero means and unit variances, and $\rho$ increasing by $0.2$ from $-1$ to $1$.}
\label{rrxy}
\centering
\scalebox{0.8}{
\begin{tabular}{ccccccccccccc}
\hline
 & (1) & (2) & (3) & (4) &Total &  & (1) & (2) & (3) & (4) &Total \\ \hline
& &\multicolumn{2}{c}{$\rho = 1.0$ } & & &                          & & \multicolumn{2}{c}{$\rho = -0.2$} & &\\
(1) & 0.2500 & 0.0000 & 0.0000 & 0.0000 & 0.2500 &                 (1) & 0.0432 & 0.0563 & 0.0668 & 0.0837 & 0.2500 \\
(2) & 0.0000 & 0.2500 & 0.0000 & 0.0000 & 0.2500 &                 (2) & 0.0563 & 0.0621 & 0.0648 & 0.0668 & 0.2500 \\
(3) & 0.0000 & 0.0000 & 0.2500 & 0.0000 & 0.2500 &                 (3) & 0.0668 & 0.0648 & 0.0621 & 0.0563 & 0.2500 \\
(4) & 0.0000 & 0.0000 & 0.0000 & 0.2500 & 0.2500 &                (4) & 0.0837 & 0.0668 & 0.0563 & 0.0432 & 0.2500 \\
Total & 0.2500 & 0.2500 & 0.2500  & 0.2500 & 1.0000 &             Total & 0.2500 & 0.2500 & 0.2500 & 0.2500 & 1.0000 \\
& &\multicolumn{2}{c}{$\rho = 0.8$ } & & &                          & & \multicolumn{2}{c}{$\rho = -0.4$} & &\\
(1) & 0.1691 & 0.0629 & 0.0164 & 0.0016 & 0.2500 &         (1) & 0.0258 & 0.0477 & 0.0692 & 0.1072 & 0.2500 \\
(2) & 0.0629 & 0.1027 & 0.0680 & 0.0164 & 0.2500 &         (2) & 0.0477 & 0.0632 & 0.0698 & 0.0692 & 0.2500 \\
(3) & 0.0164 & 0.0680 & 0.1027 & 0.0629 & 0.2500 &         (3) & 0.0692 & 0.0698 & 0.0632 & 0.0477 & 0.2500 \\
(4) & 0.0016 & 0.0164 & 0.0629 & 0.1691 & 0.2500 &         (4) & 0.1072 & 0.0692 & 0.0477 & 0.0258 & 0.2500 \\
Total & 0.2500 & 0.2500 & 0.2500 & 0.2500 & 1.0000 &      Total & 0.2500 & 0.2500 & 0.2500 & 0.2500 & 1.0000 \\
& &\multicolumn{2}{c}{$\rho = 0.6$ } & & &                          & & \multicolumn{2}{c}{$\rho = -0.6$} & &\\
(1) & 0.1345 & 0.0691 & 0.0353 & 0.0111 & 0.2500 &         (1) & 0.0111 & 0.0353 & 0.0691 & 0.1345 & 0.2500 \\
(2) & 0.0691 & 0.0797 & 0.0659 & 0.0353 & 0.2500 &         (2) & 0.0353 & 0.0659 & 0.0797 & 0.0691 & 0.2500 \\
(3) & 0.0353 & 0.0659 & 0.0797 & 0.0691 & 0.2500 &         (3) & 0.0691 & 0.0797 & 0.0659 & 0.0353 & 0.2500 \\
(4) & 0.0111 & 0.0353 & 0.0691 & 0.1345 & 0.2500 &         (4) & 0.1345 & 0.0691 & 0.0353 & 0.0111 & 0.2500 \\
Total & 0.2500 & 0.2500 & 0.2500 & 0.2500 & 1.0000 &      Total & 0.2500 & 0.2500 & 0.2500 & 0.2500 & 1.0000 \\
& &\multicolumn{2}{c}{$\rho = 0.4$ } & & &                          & & \multicolumn{2}{c}{$\rho = -0.8$} & &\\
(1) & 0.1072 & 0.0692 & 0.0477 & 0.0258 & 0.2500 &         (1) & 0.0016 & 0.0164 & 0.0629 & 0.1691 & 0.2500 \\
(2) & 0.0692 & 0.0698 & 0.0632 & 0.0477 & 0.2500 &         (2) & 0.0164 & 0.0680 & 0.1027 & 0.0629 & 0.2500 \\
(3) & 0.0477 & 0.0632 & 0.0698 & 0.0692 & 0.2500 &         (3) & 0.0629 & 0.1027 & 0.0680 & 0.0164 & 0.2500 \\
(4) & 0.0258 & 0.0477 & 0.0692 & 0.1072 & 0.2500 &         (4) & 0.1691 & 0.0629 & 0.0164 & 0.0016 & 0.2500 \\
Total & 0.2500 & 0.2500 & 0.2500 & 0.2500 & 1.0000 &      Total & 0.2500 & 0.2500 & 0.2500 & 0.2500 & 1.0000 \\
& &\multicolumn{2}{c}{$\rho = 0.2$ } & & &                          & & \multicolumn{2}{c}{$\rho = -1.0$} & &\\
(1) & 0.0837 & 0.0668 & 0.0563 & 0.0432 & 0.2500 &         (1) & 0.0000 & 0.0000 & 0.0000 & 0.2500 & 0.2500 \\
(2) & 0.0668 & 0.0648 & 0.0621 & 0.0563 & 0.2500 &         (2) & 0.0000 & 0.0000 & 0.2500 & 0.0000 & 0.2500 \\
(3) & 0.0563 & 0.0621 & 0.0648 & 0.0668 & 0.2500 &         (3) & 0.0000 & 0.2500 & 0.0000 & 0.0000 & 0.2500 \\
(4) & 0.0432 & 0.0563 & 0.0668 & 0.0837 & 0.2500 &         (4) & 0.2500 & 0.0000 & 0.0000 & 0.0000 & 0.2500 \\
Total & 0.2500 & 0.2500 & 0.2500 & 0.2500 & 1.0000 &      Total & 0.2500 & 0.2500 & 0.2500 & 0.2500 & 1.0000 \\
& &\multicolumn{2}{c}{$\rho = 0$ }  \\
(1) & 0.0625 & 0.0625 & 0.0625 & 0.0625 & 0.2500 & \\
(2) & 0.0625 & 0.0625 & 0.0625 & 0.0625 & 0.2500 & \\
(3) & 0.0625 & 0.0625 & 0.0625 & 0.0625 & 0.2500 & \\
(4) & 0.0625 & 0.0625 & 0.0625 & 0.0625 & 0.2500 & \\
Total & 0.2500 & 0.2500 & 0.2500 & 0.2500 & 1.0000 & \\  \hline
\end{tabular}
}
\end{table}

Table \ref{rrxz} presents the values of the measures $V^2_{t(f)}$ ($t=1,2,3$). 
Notably, in the case of $r \times r$ artificial contingency tables, each of $\{p_{i \cdot}\}$ and $\{p_{\cdot j}\}$ is constant, so that $V^2_{1(f)} = V^2_{2(f)} = V^2_{3(f)}$.
Table \ref{rrxz} shows that when the correlation is away from $0$, $\hat{V}^2_{t(f)}$ are close to $1$. 
Besides, $\rho = 0$ if and only if the measures show that there is a structure of null association in the table, and $\rho = \pm{1.0}$  if and only if the measures confirm that there is a structure of complete association.
Furthermore, $V^2_{t(f)}$ are invariant under any permutation of the categories.
Therefore, when the absolute values of $\rho$ are equal, $V^2_{t(f)}$ are also equal.
\begin{table}[hbtp]
\caption{Value of the measures $V^2_{t(f)}$ $(t =1, 2, 3)$ setting (a) power-divergence for any $\lambda$ and (b) $\theta$-divergence for any $\theta$.}
\label{rrxz}
\centering
\scalebox{0.9}{
\begin{tabular}{ccccccccccccc}
& \multicolumn{11}{c}{(a) $V^2_{t(f)}$ $(t =1, 2, 3)$ applying power-divergence} &  \\
\hline
& \multicolumn{11}{c}{the correlation coefficient $\rho$} & \\ \cline{2-12}
$\lambda$ & $-1.0$ & $-0.8$ & $-0.6$ & $-0.4$ & $-0.2$ & 0.0 & 0.2 & 0.4 & 0.6 & 0.8 & 1.0 \\ \hline
0.0 & 1.000 & 0.254 & 0.116 & 0.046 & 0.011 & 0.000 & 0.011 & 0.046 & 0.116 & 0.254 & 1.000 \\
0.2 & 1.000 & 0.255 & 0.118 & 0.048 & 0.011 & 0.000 & 0.011 & 0.048 & 0.118 & 0.255 & 1.000 \\
0.4 & 1.000 & 0.251 & 0.117 & 0.048 & 0.011 & 0.000 & 0.011 & 0.048 & 0.117 & 0.251 & 1.000 \\
0.6 & 1.000 & 0.244 & 0.114 & 0.046 & 0.011 & 0.000 & 0.011 & 0.046 & 0.114 & 0.244 & 1.000 \\
0.8 & 1.000 & 0.234 & 0.109 & 0.044 & 0.011 & 0.000 & 0.011 & 0.044 & 0.109 & 0.234 & 1.000 \\
1.0 & 1.000 & 0.224 & 0.103 & 0.042 & 0.010 & 0.000 & 0.010 & 0.042 & 0.103 & 0.224 & 1.000 \\  \hline
\end{tabular}
}

\scalebox{0.9}{
\begin{tabular}{ccccccccccccc}
\\ 
& \multicolumn{11}{c}{(b) $V^2_{t(f)}$ $(t =1, 2, 3)$ applying $\theta$-divergence } &  \\
\hline
& \multicolumn{11}{c}{the correlation coefficient $\rho$} & \\ \cline{2-12}
$\theta$ & $-1.0$ & $-0.8$ & $-0.6$ & $-0.4$ & $-0.2$ & $0.0$ & $0.2$ & $0.4$ & $0.6$ & $0.8$ & $1.0$ \\ \hline
0.0 & 1.000 & 0.224 & 0.103 & 0.042 & 0.010 & 0.000 & 0.010 & 0.042 & 0.103 & 0.224 & 1.000 \\
0.1 & 1.000 & 0.249 & 0.118 & 0.049 & 0.012 & 0.000 & 0.012 & 0.049 & 0.118 & 0.249 & 1.000 \\
0.3 & 1.000 & 0.278 & 0.134 & 0.055 & 0.013 & 0.000 & 0.013 & 0.055 & 0.134 & 0.278 & 1.000 \\
0.5 & 1.000 & 0.285 & 0.134 & 0.053 & 0.013 & 0.000 & 0.013 & 0.053 & 0.134 & 0.285 & 1.000 \\
0.7 & 1.000 & 0.266 & 0.115 & 0.042 & 0.010 & 0.000 & 0.010 & 0.042 & 0.115 & 0.266 & 1.000 \\
0.9 & 1.000 & 0.191 & 0.061 & 0.019 & 0.004 & 0.000 & 0.004 & 0.019 & 0.061 & 0.191 & 1.000 \\   \hline
\end{tabular}
}

\end{table}

We evaluate the performance of the approximate confidence interval for the proposed measures by the covering probability.
Suppose that $4 \times 4$ contingency tables with a sample size of $5000$ are generated by a multinomial random number based on the probability distribution in Table \ref{rrxy}.
The number of iterations is $100,000$.
This experiment enables us to evaluate whether the confidence interval for the proposed measures tends to change with the strength of the association in the contingency table.

Table \ref{rrxz_cov} shows the coverage probability of the approximate $95\%$ confidence interval for $V^2_{t(f)}$.
Note, when $\rho = \pm{1.0}$, the diagonal or antidiagonal component of the contingency table is non-zero and the probability of the other ($i, j $) cells is zero, as shown in Table \ref{rrxy}, so that $V^2_{t(f)} = 1$ as can be seen from Table \ref{rrxz}, and the confidence intervals are $[1,1]$.
Therefore, we can see that it is reasonable that the coverage probabilities are $1.000$ when $\rho = \pm{1.0}$ as shown in Table \ref{rrxz_cov}.
Another finding is that the coverage probabilities of the $95\%$ confidence interval for all $V^2_{t(f)}$, regardless of function or parameter, exceed $0.95$.
This indicates that when the sample size is sufficiently large, the performance of the confidence intervals is good regardless of the strength of association in contingency tables.

\begin{table}[hbtp]
\caption{The coverage probability of the $95\%$ confidence interval for the measures $V^2_{t(f)}$ $(t =1, 2, 3)$ setting (a) power-divergence for any $\lambda$ and (b) $\theta$-divergence for any $\theta$.}
\label{rrxz_cov}
\centering
\scalebox{0.9}{
\begin{tabular}{rcrrrrrrrrrrr}
& \multicolumn{11}{c}{(a) $V^2_{t(f)}$ $(t =1, 2, 3)$ applying power-divergence} &  \\
\hline
\multicolumn{1}{c}{\bfseries }&\multicolumn{1}{c}{\bfseries }&\multicolumn{11}{c}{the correlation coefficient $\rho$}\tabularnewline
\cline{3-13}
\multicolumn{1}{c}{$\lambda$}&\multicolumn{1}{c}{}&\multicolumn{1}{c}{$-1.0$}&\multicolumn{1}{c}{$-0.8$}&\multicolumn{1}{c}{$-0.6$}&\multicolumn{1}{c}{$-0.4$}&\multicolumn{1}{c}{$-0.2$}&\multicolumn{1}{c}{$0.0$}&\multicolumn{1}{c}{$0.2$}&\multicolumn{1}{c}{$0.4$}&\multicolumn{1}{c}{$0.6$}&\multicolumn{1}{c}{$0.8$}&\multicolumn{1}{c}{$1.0$}\tabularnewline
\hline
$0.0$&&$1.000$&$0.956$&$0.960$&$0.964$&$0.975$&$0.999$&$0.976$&$0.962$&$0.959$&$0.955$&$1.000$\tabularnewline
$0.2$&&$1.000$&$0.958$&$0.959$&$0.965$&$0.967$&$0.999$&$0.968$&$0.964$&$0.958$&$0.956$&$1.000$\tabularnewline
$0.4$&&$1.000$&$0.957$&$0.958$&$0.965$&$0.965$&$0.999$&$0.966$&$0.964$&$0.957$&$0.957$&$1.000$\tabularnewline
$0.6$&&$1.000$&$0.958$&$0.960$&$0.960$&$0.972$&$0.999$&$0.972$&$0.959$&$0.959$&$0.957$&$1.000$\tabularnewline
$0.8$&&$1.000$&$0.957$&$0.959$&$0.960$&$0.977$&$0.999$&$0.977$&$0.959$&$0.958$&$0.957$&$1.000$\tabularnewline
$1.0$&&$1.000$&$0.957$&$0.959$&$0.966$&$0.975$&$0.999$&$0.976$&$0.965$&$0.959$&$0.957$&$1.000$\tabularnewline
\hline
\end{tabular}
}

\scalebox{0.9}{
\begin{tabular}{rcrrrrrrrrrrr}
\\
& \multicolumn{11}{c}{(b) $V^2_{t(f)}$ $(t =1, 2, 3)$ applying $\theta$-divergence } &  \\
\hline
\multicolumn{1}{c}{\bfseries }&\multicolumn{1}{c}{\bfseries }&\multicolumn{11}{c}{the correlation coefficient $\rho$}\tabularnewline
\cline{3-13}
\multicolumn{1}{c}{$\theta$}&\multicolumn{1}{c}{}&\multicolumn{1}{c}{$-1.0$}&\multicolumn{1}{c}{$-0.8$}&\multicolumn{1}{c}{$-0.6$}&\multicolumn{1}{c}{$-0.4$}&\multicolumn{1}{c}{$-0.2$}&\multicolumn{1}{c}{$0.0$}&\multicolumn{1}{c}{$0.2$}&\multicolumn{1}{c}{$0.4$}&\multicolumn{1}{c}{$0.6$}&\multicolumn{1}{c}{$0.8$}&\multicolumn{1}{c}{$1.0$}\tabularnewline
\hline
$0.0$&&$1.000$&$0.956$&$0.955$&$0.960$&$0.967$&$0.987$&$0.966$&$0.959$&$0.954$&$0.955$&$1.000$\tabularnewline
$0.1$&&$1.000$&$0.956$&$0.955$&$0.959$&$0.968$&$0.982$&$0.969$&$0.958$&$0.955$&$0.954$&$1.000$\tabularnewline
$0.3$&&$1.000$&$0.954$&$0.954$&$0.954$&$0.956$&$0.979$&$0.956$&$0.954$&$0.954$&$0.953$&$1.000$\tabularnewline
$0.5$&&$1.000$&$0.952$&$0.955$&$0.953$&$0.967$&$0.981$&$0.967$&$0.953$&$0.953$&$0.953$&$1.000$\tabularnewline
$0.7$&&$1.000$&$0.953$&$0.953$&$0.956$&$0.973$&$0.988$&$0.974$&$0.956$&$0.953$&$0.952$&$1.000$\tabularnewline
$0.9$&&$1.000$&$0.928$&$0.958$&$0.971$&$0.987$&$1.000$&$0.988$&$0.970$&$0.959$&$0.928$&$1.000$\tabularnewline
\hline
\end{tabular}
}

\end{table}

\clearpage
\subsection*{Numerical study 2}
In numerical experiment 2, we investigate the relationship between the measures when the number of rows and columns is varied.
The artificial data is generated from a bivariate normal distribution with $ \rho = 0.4 $, and the number of rows and columns is increased to $4$, $8$, and $12$, respectively.
Table \ref{rc} is the result of the measures $V^2_{1(f)}$, $V^2_{2(f)}$, and $V^2_{3(f)}$ (especially $V^2_{H(f)}$ and $V^2_{G(f)}$). 
In Table \ref{rc}, as the number of columns increase, $\hat{V}^2_{1(f)}$ increases but $\hat{V}^2_{2(f)}$ decreases for each $\lambda$ and $\theta$. 
On the other hand, for each $\lambda$ and $\theta$, as the number of rows increases, $\hat{V}^2_{1(f)}$ decreases but $\hat{V}^2_{2(f)}$ increases.
The increase or decrease in these values is attributed to the increase in the number of explanatory variables, which makes it easier to capture the explanatory variables.
The measure $V^2_{3(f)}$ combines both measures $V^2_{1(f)}$ and $V^2_{2(f)}$ to the extent to which knowledge of the value of one variable can help us predict the value of the other variable. 
Therefore, for each $\lambda$ and $\theta$, as the number of rows or columns increases, $V^2_{H(f)}$ and $V^2_{G(f)}$ are less, and the values remain the same even if the number of rows and columns are interchanged.
In addition, closer inspection of Table \ref{rc} shows that the $\theta$-divergence type measures is less affected by the number of rows and columns than the power-divergence type measures.
Although this finding may be somewhat influenced by the discrete bivariate normal distribution, it may be better to use the $\theta$-divergence type meausres with $\theta \geq 0.5$ in a contingency table where one side of the number of rows and columns is large.

\begin{table}[hbtp]
\caption{Value of the measures $V^2_{t(f)}$ $(t =1, 2, 3)$ setting (a) power-divergence for any $\lambda$ and (b) $\theta$-divergence for any $\theta$ in some artificial data with $\rho = 0.4$.}
\label{rc}
\centering
\begin{multicols}{2}
\scalebox{0.90}{
\begin{tabular}{cccccc}
\multicolumn{6}{c}{(a) applying power-divergence}   \\
\hline
$\lambda$ & $r \times c$ & $V^2_{1(f)}$ & $V^2_{2(f)}$ & $V^2_{H(f)}$ & $V^2_{G(f)}$ \\ \hline
0.0 & $4 \times 4$ & 0.046 & 0.046 & 0.046 & 0.046  \\
      & $4 \times 8$ & 0.051 & 0.034 & 0.041 & 0.041 \\
      & $8 \times 4$ & 0.034 & 0.051 & 0.041 & 0.041 \\ 
      & $4 \times 12$ & 0.052 & 0.029 & 0.037 & 0.039 \\
      & $12 \times 4$ & 0.029 & 0.052 & 0.037 & 0.039 \\ \hline

0.2 & $4 \times 4$ & 0.048 & 0.048 & 0.048 & 0.048  \\
      & $4 \times 8$ & 0.052 & 0.032 & 0.040 & 0.041 \\
      & $8 \times 4$ & 0.034 & 0.052 & 0.040 & 0.041 \\ 
      & $4 \times 12$ & 0.054 & 0.027 & 0.035 & 0.038 \\
      & $12 \times 4$ & 0.027 & 0.054 & 0.035 & 0.038 \\ \hline

0.4 & $4 \times 4$ & 0.048 & 0.048 & 0.048 & 0.048  \\
      & $4 \times 8$ & 0.052 & 0.030 & 0.038 & 0.040 \\
      & $8 \times 4$ & 0.030 & 0.052 & 0.038 & 0.040 \\ 
      & $4 \times 12$ & 0.054 & 0.023 & 0.032 & 0.035 \\
      & $12 \times 4$ & 0.023 & 0.054 & 0.032 & 0.035 \\ \hline
	  
0.6 & $4 \times 4$ & 0.046 & 0.046 & 0.046 & 0.046  \\
      & $4 \times 8$ & 0.051 & 0.027 & 0.035 & 0.037 \\
      & $8 \times 4$ & 0.027 & 0.051 & 0.035 & 0.037 \\
      & $4 \times 12$ & 0.052 & 0.020 & 0.029 & 0.032 \\
      & $12 \times 4$ & 0.020 & 0.052 & 0.029 & 0.032 \\ \hline
      
0.8 & $4 \times 4$ & 0.044 & 0.044 & 0.044 & 0.044  \\
      & $4 \times 8$ & 0.049 & 0.023 & 0.032 & 0.034 \\
      & $8 \times 4$ & 0.023 & 0.049 & 0.032 & 0.034 \\ 
      & $4 \times 12$ & 0.050 & 0.016 & 0.025 & 0.029 \\
      & $12 \times 4$ & 0.016 & 0.050 & 0.025 & 0.029 \\ \hline

1.0 & $4 \times 4$ & 0.042 & 0.042 & 0.042 & 0.042  \\
      & $4 \times 8$ & 0.046 & 0.020 & 0.028 & 0.030 \\
      & $8 \times 4$ & 0.020 & 0.046 & 0.028 & 0.030 \\
      & $4 \times 12$ & 0.048 & 0.013 & 0.020 & 0.025 \\
      & $12 \times 4$ & 0.013 & 0.048 & 0.020 & 0.025 \\ \hline
\end{tabular}
}

\scalebox{0.90}{
\begin{tabular}{ccccccccccccc}
\multicolumn{6}{c}{(b) applying $\theta$-divergence}   \\
\hline
$\theta$ & $r \times c$ & $V^2_{1(f)}$ & $V^2_{2(f)}$ & $V^2_{H(f)}$ & $V^2_{G(f)}$ \\ \hline
0.0 & $4 \times 4$ & 0.042 & 0.042 & 0.042 & 0.042  \\
      & $4 \times 8$ & 0.046 & 0.020 & 0.028 & 0.030 \\
      & $8 \times 4$ & 0.020 & 0.046 & 0.028 & 0.030 \\
      & $4 \times 12$ & 0.048 & 0.013 & 0.020 & 0.025 \\
      & $12 \times 4$ & 0.013 & 0.048 & 0.020 & 0.025 \\ \hline

0.1 & $4 \times 4$ & 0.049 & 0.049 & 0.049 & 0.049  \\
      & $4 \times 8$ & 0.053 & 0.030 & 0.038 & 0.040 \\
      & $8 \times 4$ & 0.030 & 0.053 & 0.038 & 0.040 \\
      & $4 \times 12$ & 0.055 & 0.024 & 0.033 & 0.036 \\
      & $12 \times 4$ & 0.024 & 0.055 & 0.033 & 0.036 \\ \hline
	  
0.3 & $4 \times 4$ & 0.055 & 0.055 & 0.055 & 0.055  \\
      & $4 \times 8$ & 0.060 & 0.042 & 0.050 & 0.050 \\
      & $8 \times 4$ & 0.042 & 0.060 & 0.050 & 0.050 \\ 
      & $4 \times 12$ & 0.062 & 0.038 & 0.047 & 0.048 \\
      & $12 \times 4$ & 0.038 & 0.062 & 0.047 & 0.048 \\ \hline
      	  
0.5 & $4 \times 4$ & 0.053 & 0.053 & 0.053 & 0.053  \\
      & $4 \times 8$ & 0.059 & 0.045 & 0.051 & 0.051 \\
      & $8 \times 4$ & 0.045 & 0.059 & 0.051 & 0.051 \\ 
      & $4 \times 12$ & 0.060 & 0.042 & 0.050 & 0.050 \\
      & $12 \times 4$ & 0.042 & 0.060 & 0.050 & 0.050 \\ \hline

0.7 & $4 \times 4$ & 0.042 & 0.042 & 0.042 & 0.042  \\
      & $4 \times 8$ & 0.047 & 0.038 & 0.042 & 0.042 \\
      & $8 \times 4$ & 0.038 & 0.047 & 0.042 & 0.042 \\
      & $4 \times 12$ & 0.048 & 0.037 & 0.042 & 0.042 \\
      & $12 \times 4$ & 0.037 & 0.048 & 0.042 & 0.042 \\ \hline
	        
0.9 & $4 \times 4$ & 0.019 & 0.019 & 0.019 & 0.019  \\
      & $4 \times 8$ & 0.022 & 0.018 & 0.020 & 0.020 \\
      & $8 \times 4$ & 0.018 & 0.022 & 0.020 & 0.020 \\ 
      & $4 \times 12$ & 0.022 & 0.018 & 0.020 & 0.020 \\
      & $12 \times 8$ & 0.018 & 0.022 & 0.020 & 0.020 \\ \hline
\end{tabular}
}
\end{multicols}
\end{table}

We evaluate the performance of the approximate confidence interval for the proposed measures by the covering probability.
Suppose that, contingency tables with a sample size of 5000, varying the number of rows and columns, are generated by a multinomial random number based on the probability distribution obtained from a discrete bivariate normal distribution.
The number of iterations is $100,000$.
This experiment enables us to evaluate whether the influence of the number of rows and columns in the contingency table tends to change the confidence interval for the proposed measures.

Table \ref{rc_cov} shows the coverage probability of the approximate $95\%$ confidence interval for $V^2_{t(f)}$.
The important finding is that the coverage probabilities of the $95\%$ confidence interval for all $V^2_{t(f)}$, regardless of function or parameter, exceed $0.91$.
This indicates that when the sample size is sufficiently large, the performance of the confidence intervals is good regardless of the number of rows and columns in contingency tables.

\begin{table}[hbtp]
\caption{The coverage probability of the $95\%$ confidence interval for the measures $V^2_{t(f)}$ $(t =1, 2, 3)$ setting (a) power-divergence for any $\lambda$ and (b) $\theta$-divergence for any $\theta$ in some artificial data with $\rho = 0.4$.}
\label{rc_cov}
\centering
\begin{multicols}{2}
\scalebox{0.90}{
\begin{tabular}{cccccc}
\multicolumn{6}{c}{(a) applying power-divergence}   \\
\hline
$\lambda$ & $r \times c$ & $V^2_{1(f)}$ & $V^2_{2(f)}$ & $V^2_{H(f)}$ & $V^2_{G(f)}$ \\ \hline
0.0 & $4 \times 4$ & 0.962 & 0.962 & 0.962 & 0.962  \\
      & $4 \times 8$ & 0.959 & 0.964 & 0.964 & 0.948 \\
      & $8 \times 4$ & 0.965 & 0.959 & 0.965 & 0.949 \\ 
      & $4 \times 12$ & 0.938 & 0.950 & 0.934 & 0.949 \\
      & $12 \times 4$ & 0.949 & 0.938 & 0.935 & 0.948 \\ \hline

0.2 & $4 \times 4$ & 0.964 & 0.964 & 0.964 & 0.964  \\
      & $4 \times 8$ & 0.948 & 0.949 & 0.958 & 0.954 \\
      & $8 \times 4$ & 0.951 & 0.950 & 0.959 & 0.955 \\ 
      & $4 \times 12$ & 0.947 & 0.966 & 0.917 & 0.952 \\
      & $12 \times 4$ & 0.966 & 0.946 & 0.917 & 0.951 \\ \hline

0.4 & $4 \times 4$ & 0.964 & 0.964 & 0.964 & 0.964  \\
      & $4 \times 8$ & 0.948 & 0.966 & 0.958 & 0.965 \\
      & $8 \times 4$ & 0.966 & 0.950 & 0.959 & 0.966 \\ 
      & $4 \times 12$ & 0.947 & 0.934 & 0.920 & 0.940 \\
      & $12 \times 4$ & 0.934 & 0.946 & 0.920 & 0.928 \\ \hline
	  
0.6 & $4 \times 4$ & 0.959 & 0.959 & 0.959 & 0.959  \\
      & $4 \times 8$ & 0.952 & 0.971 & 0.959 & 0.961 \\
      & $8 \times 4$ & 0.972 & 0.953 & 0.959 & 0.961 \\
      & $4 \times 12$ & 0.925 & 0.972 & 0.962 & 0.940 \\
      & $12 \times 4$ & 0.972 & 0.925 & 0.962 & 0.939 \\ \hline
      
0.8 & $4 \times 4$ & 0.959 & 0.959 & 0.959 & 0.959  \\
      & $4 \times 8$ & 0.955 & 0.959 & 0.969 & 0.965 \\
      & $8 \times 4$ & 0.959 & 0.955 & 0.970 & 0.966 \\ 
      & $4 \times 12$ & 0.929 & 0.953 & 0.972 & 0.966 \\
      & $12 \times 4$ & 0.952 & 0.929 & 0.971 & 0.966 \\ \hline

1.0 & $4 \times 4$ & 0.965 & 0.965 & 0.965 & 0.965  \\
      & $4 \times 8$ & 0.949 & 0.976 & 0.969 & 0.954 \\
      & $8 \times 4$ & 0.977 & 0.951 & 0.970 & 0.955 \\
      & $4 \times 12$ & 0.951 & 0.977 & 0.935 & 0.962 \\
      & $12 \times 4$ & 0.977 & 0.951 & 0.935 & 0.962 \\ \hline
\end{tabular}
}

\scalebox{0.90}{
\begin{tabular}{ccccccccccccc}
\multicolumn{6}{c}{(b) applying $\theta$-divergence}   \\
\hline
$\theta$ & $r \times c$ & $V^2_{1(f)}$ & $V^2_{2(f)}$ & $V^2_{H(f)}$ & $V^2_{G(f)}$ \\ \hline
0.0 & $4 \times 4$ & 0.965 & 0.965 & 0.965 & 0.965  \\
      & $4 \times 8$ & 0.949 & 0.976 & 0.969 & 0.954 \\
      & $8 \times 4$ & 0.977 & 0.951 & 0.970 & 0.955 \\
      & $4 \times 12$ & 0.951 & 0.977 & 0.935 & 0.962 \\
      & $12 \times 4$ & 0.977 & 0.951 & 0.935 & 0.962 \\ \hline

0.1 & $4 \times 4$ & 0.964 & 0.964 & 0.964 & 0.964  \\
      & $4 \times 8$ & 0.948 & 0.966 & 0.952 & 0.960 \\
      & $8 \times 4$ & 0.967 & 0.950 & 0.953 & 0.961 \\
      & $4 \times 12$ & 0.947 & 0.952 & 0.927 & 0.934 \\
      & $12 \times 4$ & 0.952 & 0.946 & 0.927 & 0.962 \\ \hline
	  
0.3 & $4 \times 4$ & 0.960 & 0.960 & 0.960 & 0.960  \\
      & $4 \times 8$ & 0.947 & 0.953 & 0.960 & 0.945 \\
      & $8 \times 4$ & 0.954 & 0.948 & 0.961 & 0.947 \\ 
      & $4 \times 12$ & 0.942 & 0.941 & 0.936 & 0.922 \\
      & $12 \times 4$ & 0.941 & 0.942 & 0.936 & 0.922 \\ \hline
      	  
0.5 & $4 \times 4$ & 0.957 & 0.957 & 0.957 & 0.957  \\
      & $4 \times 8$ & 0.958 & 0.952 & 0.954 & 0.945 \\
      & $8 \times 4$ & 0.953 & 0.958 & 0.955 & 0.946 \\ 
      & $4 \times 12$ & 0.936 & 0.920 & 0.944 & 0.921 \\
      & $12 \times 4$ & 0.920 & 0.936 & 0.944 & 0.921 \\ \hline

0.7 & $4 \times 4$ & 0.960 & 0.960 & 0.960 & 0.960  \\
      & $4 \times 8$ & 0.954 & 0.950 & 0.951 & 0.946 \\
      & $8 \times 4$ & 0.951 & 0.955 & 0.953 & 0.947 \\
      & $4 \times 12$ & 0.928 & 0.943 & 0.943 & 0.930 \\
      & $12 \times 4$ & 0.943 & 0.927 & 0.943 & 0.930 \\ \hline
	        
0.9 & $4 \times 4$ & 0.975 & 0.975 & 0.975 & 0.975  \\
      & $4 \times 8$ & 0.974 & 0.968 & 0.975 & 0.974 \\
      & $8 \times 4$ & 0.969 & 0.974 & 0.975 & 0.974 \\ 
      & $4 \times 12$ & 0.945 & 0.965 & 0.964 & 0.960 \\
      & $12 \times 8$ & 0.966 & 0.945 & 0.965 & 0.961 \\ \hline
\end{tabular}
}
\end{multicols}
\end{table}

\clearpage
\section{Additional examples}
In this Appendix D, we show tables of the results of some actual data examples on measures by function and parameter. 
We use the Tomizawa's power-divergence type measures ($f(x)=(x^{\lambda+1}-x)/\lambda(\lambda+1)$ for $\lambda = 0.0, 0.6, 1.0, 1.2, 1.5$) and the newly proposed the $\theta$-divergence type measures ($f(x) = (x-1)^2/(\theta x + 1 - \theta) + (x-1)/(1 - \theta)$ for $\theta = 0.0, 0.3, 0.5, 0.7, 0.9$), both of which are a single-parameter divergence and extensions of the Cram\'er's coefficient $V^2$.
\subsection*{Example 1}
Consider the data in Table \ref{alcohol}, taken from \cite{andersen1994statistical}.
These are data from a Danish Welfare Study which describes the cross-classification of alcohol consumption and social rank. 
Alcohol consumption in the contingency table is grouped according to the number of "units" consumed per day. 
A unit is typically a beer, half a bottle of wine, or 2cl or 40$\%$ alcohol. 
This data can be assumed that row variables are response variables and column variables are explanatory variables.
By applying the measure $V^2_{1(f)}$, we consider to what degree the prediction of alcohol consumption can be improved when the social rank of an individual is known.
\begin{table}[htbp]
\caption{Association between alcohol consumption and social rank}
\centering
\label{alcohol}
\begin{tabular}{cccccc}
\hline
Units of alcohol & \multicolumn{4}{c}{Social rank group} & \\ \cline{2-5}
consumer per day & $\rm{I}-\rm{II}$ & $\rm{III}$ & $\rm{IV}$ & $\rm{V}$ & Total  \\ \hline
Under 1 & 98 & 338 & 484 & 484 & 1404  \\
1-2 & 235 & 406 & 588 & 385 & 1614 \\
More than 2 & 123 & 144 & 191 & 137 & 595 \\   \hline
Total  & 456 & 888 & 1263 & 1006  & 3613 \\ \hline
\multicolumn{2}{l}{Source: \cite{andersen1994statistical}}  \\
\end{tabular}
\end{table}

Table \ref{alcoholr} shows the estimates of the measure, standard errors, and confidence intervals. 
The confidence intervals for all $V^2_{1(f)}$ do not contain zero for all $\lambda$ and $\theta$. 
This shows that there is a rather weak association between alcohol alcohol consumption and social rank, including KL-divergence ($\lambda=0$), Pearson divergence ($\lambda=1, \theta=0$), triangular discrimination ($\theta=0.5$), and other perspectives.
Another important finding is that, for example, when $\lambda=1.0$ and $\theta = 0.0$, $\hat{V}^2_{1(f)}$ indicates that the strength of association between alcohol consumption and social rank is estimated to be 0.015 times the complete association. 
Hence, while predicting the value of alcohol consumption of an individual, we can predict it $1.5\%$ better than when we do not. 

\begin{table}[htbp]
\caption{Estimate of the measure $V^2_{1(f)}$, estimated approximate standard error for $\hat{V}^2_{1(f)}$, and approximate $95\%$ confidence interval of $V^2_{1(f)}$, applying (a) power-divergence for any $\lambda$ and (b) $\theta$-divergence for any $\theta$.}
\centering
\begin{multicols}{2}
\label{alcoholr}
\begin{tabular}{cccc}
\multicolumn{4}{c}{(a) $\hat{V}^2_{1(f)}$ applying power-divergence}   \\
\hline
$\lambda$ & $\hat{V}^2_{1(f)}$ & SE & $95\%$CI \\  \hline
0.0 & 0.015 & 0.003 & (0.010, 0.020) \\
0.2 & 0.016 & 0.003 & (0.010, 0.022) \\
0.4 & 0.016 & 0.003 & (0.010, 0.022) \\
0.6 & 0.016 & 0.003 & (0.010, 0.022) \\
0.8 & 0.016 & 0.003 & (0.010, 0.022) \\
1.0 & 0.015 & 0.003 & (0.010, 0.021) \\
1.2 & 0.014 & 0.003 & (0.009, 0.020) \\
1.5 & 0.013 & 0.002 & (0.008, 0.017) \\ \hline
\end{tabular}

\begin{tabular}{cccc}
\multicolumn{4}{c}{(b) $\hat{V}^2_{1(f)}$ applying $\theta$-divergence}  \\
\hline
$\theta$ & $\hat{V}^2_{1(f)}$ & SE & $95\%$CI \\  \hline
0.0 & 0.015 & 0.003 & (0.010, 0.021) \\
0.1 & 0.017 & 0.003 & (0.011, 0.024) \\
0.2 & 0.018 & 0.003 & (0.012, 0.025) \\
0.3 & 0.018 & 0.003 & (0.012, 0.025) \\
0.4 & 0.018 & 0.003 & (0.011, 0.024) \\
0.5 & 0.017 & 0.003 & (0.011, 0.023) \\
0.6 & 0.015 & 0.003 & (0.009, 0.020) \\
0.7 & 0.012 & 0.002 & (0.008, 0.017) \\
0.8 & 0.009 & 0.002 & (0.006, 0.013) \\
0.9 & 0.005 & 0.001 & (0.003, 0.007) \\ \hline
\end{tabular}

\end{multicols}
\end{table}

\subsection*{Example 2}
Consider the data in Table \ref{race}, taken from the 2006 General Social Survey.
These are data, which show the relationship between family income and education in the United States separately for black and white categories of race.
This data can be assumed that row variables are response variables and column variables are explanatory variables.
By applying the measures $V^2_{1(f)}$, we consider to what educational degree can be improved when the prediction of family income for black and white categories of an individual is known.

\begin{table}[hbtp]
\caption{Data on educational degrees and family income, by race}
\centering
\label{race}
\begin{tabular}{lcccc}
\\
\multicolumn{5}{c}{(a) Black Person} 
\\
\hline
 & \multicolumn{3}{c}{Income} & \\ \cline{2-4}
Degree  & Blow Average & Average  & Above Average & Total  \\ \hline
$<$ High school & 43 & 36 & 5 & 84  \\
High school, junior college & 104 & 140 & 23 & 267 \\
College, grad school & 16 & 30 & 18 & 64 \\ \hline
Total & 163 & 206 & 46  & 415 \\ \hline
\multicolumn{2}{l}{Source: 2006 General Social Survey}  \\
\end{tabular}

\begin{tabular}{lcccc}
\\
\multicolumn{5}{c}{(b) White Person}
\\
\hline
 & \multicolumn{3}{c}{Income} & \\ \cline{2-4}
Degree  & Blow Average & Average  & Above Average & Total  \\ \hline
$<$ High school & 114 & 97 & 12 & 223  \\
High school, junior college & 410 & 658 & 221 & 1289 \\
College, grad school & 97 & 259 & 287 & 643 \\ \hline
Total & 621 & 1014 & 520  & 2155 \\ \hline
\multicolumn{2}{l}{Source: 2006 General Social Survey}  \\
\end{tabular}
\end{table}

Table \ref{race_m} shows the estimates of the measures, standard errors, and $95\%$ confidence intervals. In addition, Tables \ref{race_m}(a1, b1) and \ref{race_m}(b1, b2) show the results of the analysis of Tables \ref{race}(a) and \ref{race}(b), respectively.
One interesting finding is the confidence intervals for all $V^2_{1(f)}$ do not contain zero for any $\lambda$ and any $\theta$. 
The results show that the two actual data have an associated structure from a point of view other than the Cram\'er's coefficient $V^2$.
Another important finding is the comparison of the confidence intervals.
For conventional power-divergence type measures, a comparison of Tables \ref{race_m}(a1) and \ref{race_m}(b1) shows that confidence intervals overlap for each $\lambda$.
On the other hand, when $\theta = 0.9$ in Tables \ref{race_m}(a2) and \ref{race_m}(b2), the confidence intervals do not overlap, and Table \ref{race}(a), where the estimate is closer to $0$, has higher independence.

\begin{table}[hbtp]
\caption{Estimate of the measure $V^2_{1(f)}$, estimated approximate standard error for $\hat{V}^2_{1(f)}$, and approximate $95\%$ confidence interval of $V^2_{1(f)}$ applying (a1, b1) power-divergence for any $\lambda$ and (a2, b2) $\theta$-divergence for any $\theta$.}
\label{race_m}
\centering
\begin{multicols}{2}
\begin{tabular}{cccc}
\multicolumn{4}{c}{(a1) applying power-divergence}   \\
\hline
$\lambda$ & $\hat{V}^2_{1(f)}$ & SE & $95\%$CI \\  \hline
0.0 & 0.032 & 0.014 & (0.005, 0.059) \\
0.2 & 0.034 & 0.015 & (0.005, 0.063) \\
0.4 & 0.035 & 0.016 & (0.005, 0.066) \\
0.6 & 0.036 & 0.016 & (0.005, 0.067) \\
0.8 & 0.035 & 0.016 & (0.004, 0.066) \\
1.0 & 0.034 & 0.016 & (0.003, 0.064) \\
1.2 & 0.032 & 0.015 & (0.002, 0.061) \\
1.5 & 0.028 & 0.014 & (0.001, 0.056) \\ \hline
\end{tabular}

\begin{tabular}{cccc}
\\ 
\multicolumn{4}{c}{(a2) applying $\theta$-divergence}  \\
\hline
$\theta$ & $\hat{V}^2_{1(f)}$ & SE & $95\%$CI \\  \hline
0.0 & 0.034 & 0.016 & (0.003, 0.064) \\
0.1 & 0.038 & 0.017 & (0.005, 0.072) \\
0.2 & 0.040 & 0.017 & (0.006, 0.074) \\
0.3 & 0.040 & 0.017 & (0.007, 0.072) \\
0.4 & 0.037 & 0.016 & (0.007, 0.068) \\
0.5 & 0.034 & 0.014 & (0.006, 0.062) \\
0.6 & 0.029 & 0.012 & (0.005, 0.054) \\
0.7 & 0.024 & 0.010 & (0.004, 0.044) \\
0.8 & 0.017 & 0.008 & (0.002, 0.032) \\
0.9 & 0.009 & 0.004 & (0.001, 0.017) \\ \hline
\end{tabular}

\begin{tabular}{cccc}
\multicolumn{4}{c}{(b1) applying power-divergence}   \\
\hline
$\lambda$ & $\hat{V}^2_{1(f)}$ & SE & $95\%$CI \\  \hline
0.0 & 0.068 & 0.008 & (0.052, 0.083) \\
0.2 & 0.071 & 0.008 & (0.055, 0.087) \\
0.4 & 0.072 & 0.008 & (0.055, 0.088) \\
0.6 & 0.070 & 0.008 & (0.054, 0.086) \\
0.8 & 0.067 & 0.008 & (0.051, 0.082) \\
1.0 & 0.062 & 0.007 & (0.047, 0.076) \\
1.2 & 0.056 & 0.007 & (0.042, 0.069) \\
1.5 & 0.046 & 0.006 & (0.034, 0.057) \\ \hline
\end{tabular}

\begin{tabular}{cccc}
\\
\multicolumn{4}{c}{(b2) applying $\theta$-divergence}  \\
\hline
$\theta$ & $\hat{V}^2_{1(f)}$ & SE & $95\%$CI \\  \hline
0.0 & 0.062 & 0.007 & (0.047, 0.076) \\
0.1 & 0.076 & 0.009 & (0.059, 0.094) \\
0.2 & 0.082 & 0.010 & (0.064, 0.101) \\
0.3 & 0.084 & 0.010 & (0.065, 0.102) \\
0.4 & 0.081 & 0.009 & (0.063, 0.100) \\
0.5 & 0.076 & 0.009 & (0.059, 0.094) \\
0.6 & 0.069 & 0.008 & (0.053, 0.085) \\
0.7 & 0.058 & 0.007 & (0.044, 0.072) \\
0.8 & 0.044 & 0.006 & (0.032, 0.056) \\
0.9 & 0.026 & 0.004 & (0.018, 0.033) \\ \hline
\end{tabular}
\end{multicols}
\end{table}

\clearpage
\subsection*{Example 3}
Consider the data in Table \ref{car} taken from \cite{read1988goodness}. 
These are data on 4831 car accidents, which are cross-classified according to accident type and accident severity. 
This data can be assumed that row variables are explanatory variables and column variables are response variables.
By applying the measures $V^2_{2(f)}$, we consider the degree to which prediction of the accident severity can be improved when the accident type is known.

\begin{table}[htbp]
\caption{The 4831 car accidents}
\centering
\label{car}
\begin{tabular}{ccccc}
\hline
 & \multicolumn{3}{c}{Accident severity} & \\ \cline{2-4}
Accident & Not & Moderately & &   \\ 
type & severe & severe & Severe & Total  \\ \hline
Rollover & 2365 & 944 & 412 & 3721 \\
Not rollover & 249 & 585 & 276 & 1110 \\   \hline
Total  & 2614 & 1529 & 688 & 4831 \\ \hline
\multicolumn{2}{l}{Source: \cite{read1988goodness}}  \\
\end{tabular}
\end{table}

Table \ref{carr} gives the estimates of the measures, standard errors, and confidence intervals. Likewise, the confidence intervals for all $V^2_{2(f)}$ do not contain zero for any $\lambda$ and $\theta$. 
This shows that there is a weak association between accident type and accident severity, including KL-divergence ($\lambda=0$), Pearson divergence ($\lambda=1, \theta=0$), triangular discrimination ($\theta=0.5$), and other perspectives.
Another important finding is that, for instance, when $\lambda=1.0$ and $\theta = 0.0$, $\hat{V}^2_{2(f)}$ indicates that the strength of the association between accident type and accident severity is estimated to be 0.060 times the complete association. 
Therefore, while predicting the value of accident severity of an individual, we can predict it $6.0\%$ better than when we do not. 

\begin{table}[htbp]
\caption{Estimate of the measure $V^2_{2(f)}$, estimated approximate standard error for $\hat{V}^2_{2(f)}$, and approximate $95\%$ confidence interval of $V^2_{2(f)}$, applying (a) power-divergence for any $\lambda$ and (b) $\theta$-divergence for any $\theta$.}
\centering
\begin{multicols}{2}
\label{carr}
\begin{tabular}{cccc}
\multicolumn{4}{c}{(a) applying power-divergence}   \\
\hline
$\lambda$ & $\hat{V}^2_{2(f)}$ & SE & $95\%$CI \\  \hline
0.0 & 0.064 & 0.005 & (0.054, 0.074) \\
0.2 & 0.067 & 0.005 & (0.057, 0.077) \\
0.4 & 0.068 & 0.005 & (0.058, 0.078) \\
0.6 & 0.067 & 0.005 & (0.057, 0.077) \\
0.8 & 0.064 & 0.005 & (0.055, 0.074) \\
1.0 & 0.060 & 0.004 & (0.052, 0.069) \\
1.2 & 0.056 & 0.004 & (0.048, 0.064) \\
1.5 & 0.048 & 0.004 & (0.041, 0.055) \\ \hline
\end{tabular}

\begin{tabular}{cccc}
\multicolumn{4}{c}{(b) applying $\theta$-divergence}  \\
\hline
$\theta$ & $\hat{V}^2_{2(f)}$ & SE & $95\%$CI \\  \hline
0.0 & 0.060 & 0.004 & (0.052, 0.069) \\
0.1 & 0.071 & 0.005 & (0.061, 0.082) \\
0.2 & 0.077 & 0.006 & (0.065, 0.088) \\
0.3 & 0.078 & 0.006 & (0.067, 0.090) \\
0.4 & 0.077 & 0.006 & (0.065, 0.089) \\
0.5 & 0.073 & 0.006 & (0.062, 0.084) \\
0.6 & 0.066 & 0.005 & (0.055, 0.076) \\
0.7 & 0.056 & 0.005 & (0.047, 0.065) \\
0.8 & 0.042 & 0.004 & (0.035, 0.050) \\
0.9 & 0.024 & 0.002 & (0.020, 0.028) \\ \hline
\end{tabular}

\end{multicols}
\end{table}

\subsection*{Example 4}
Consider the data in Table \ref{eyes_UK}, taken from \cite{stuart1953estimation}. These tables are naked eye acuity data for UK men and women. In Table \ref{eyes_UK}, the row and column variables are the right and left eye grades,
respectively, with the categories ordered from the highest grade (1) to the lowest grade (4). 
These tables cannot distinguish whether the row and column variables are explanatory or response variables, so we apply the measure $V^2_{H(f)}$.
\begin{table}[hbtp]
\caption{Naked eye acuity data for UK men and women}
\centering
\label{eyes_UK}
\begin{tabular}{lccccc}
\multicolumn{6}{c}{(a) Women} 
\\
\hline
 & \multicolumn{4}{c}{Left eye grade} & \\ \cline{2-5}
Right eye grade & (1) & (2) & (3) & (4) & Total  \\ \hline
Best (1) & 1520 & 266 & 124 & 66 & 1976  \\
Second (2) & 234 & 1512 & 432 & 78 & 2256 \\
Third (3) & 117 & 362 & 1772 & 205 & 2456 \\
Worst (4)  & 36 & 82 & 179 & 492 & 789 \\   \hline
Total & 1907 & 2222 & 2507 & 841  & 7477 \\ \hline
\multicolumn{2}{l}{Source: \cite{stuart1953estimation}}  \\
\end{tabular}

\begin{tabular}{lccccc}
\\
\multicolumn{6}{c}{(b) Men} 
\\
\hline
 & \multicolumn{4}{c}{Left eye grade} & \\ \cline{2-5}
Right eye grade & (1) & (2) & (3) & (4) & Total  \\ \hline
Best (1) & 821 & 112 & 85 & 35 & 1053  \\
Second (2) & 116 & 494 & 145 & 27 & 782 \\
Third (3) & 72 & 151 & 583 & 87 & 893 \\
Worst (4)  & 43 & 34 & 106 & 331 & 514 \\   \hline
Total & 1052 & 791 & 919 & 480  & 3242 \\ \hline
\multicolumn{2}{l}{Source: \cite{stuart1953estimation}}  \\
\end{tabular}
\end{table}

Table \ref{eyesr_UK} gives the estimates of the measures, standard errors, and confidence intervals.
In addition, Tables \ref{eyesr_UK}(a1, b1) and \ref{eyesr_UK}(b1, b2) show the results of the analysis of Tables \ref{eyes_UK}(a) and \ref{eyes_UK}(b), respective.
The results of this analysis show that the two actual data have the strong structure of association in terms of the estimates and confidence intervals for all $\lambda$ and $\theta$.
As for the comparing values of the measures between Tables \ref{eyesr_UK}(a1, b1) and \ref{eyesr_UK}(b1, b2), it can be said that the strength of association between the right and the left eye is greater for women in term of the estimates for each parameter.
Comparing confidence intervals between Tables \ref{eyesr_UK}(a1, b1) and \ref{eyesr_UK}(b1, b2) can be similarly concluded, we can see that the confidence intervals overlap for all $\lambda$ and $\theta$.
These results show that the association between right and left naked eye acuity is stronger for women than for men based on the estimate of the measure, but the confidence intervals do not allow us to conclude that there is a difference in the strength of the association.

\begin{table}[hbtp]
\caption{Estimate of the measure $V^2_{H(f)}$, estimated approximate standard error for $\hat{V}^2_{H(f)}$, and approximate $95\%$ confidence interval of $V^2_{H(f)}$ applying (a1, b1) power-divergence for any $\lambda$ and (b1, b2) $\theta$-divergence for any $\theta$.}
\label{eyesr_UK}
\centering
\begin{multicols}{2}
\begin{tabular}{cccc}
\multicolumn{4}{c}{(a1) applying power-divergence}   \\
\hline
$\lambda$ & $\hat{V}^2_{H(f)}$ & SE & $95\%$CI \\  \hline
0.0 & 0.338 & 0.008 & (0.324, 0.353) \\
0.2 & 0.358 & 0.008 & (0.343, 0.373) \\
0.4 & 0.369 & 0.008 & (0.353, 0.385) \\
0.6 & 0.372 & 0.008 & (0.355, 0.388) \\
0.8 & 0.369 & 0.009 & (0.352, 0.386) \\
1.0 & 0.361 & 0.009 & (0.344, 0.378) \\
1.2 & 0.349 & 0.009 & (0.332, 0.366) \\
1.5 & 0.325 & 0.009 & (0.307, 0.343) \\ \hline
\end{tabular}

\begin{tabular}{cccc}
\\ 
\multicolumn{4}{c}{(a2) applying $\theta$-divergence}  \\
\hline
$\theta$ & $\hat{V}^2_{H(f)}$ & SE & $95\%$CI \\  \hline
0.0 & 0.361 & 0.009 & (0.344, 0.378) \\
0.1 & 0.390 & 0.008 & (0.374, 0.407) \\
0.2 & 0.397 & 0.008 & (0.381, 0.414) \\
0.3 & 0.393 & 0.008 & (0.377, 0.409) \\
0.4 & 0.380 & 0.008 & (0.365, 0.396) \\
0.5 & 0.359 & 0.008 & (0.344, 0.374) \\
0.6 & 0.328 & 0.008 & (0.313, 0.343) \\
0.7 & 0.285 & 0.007 & (0.271, 0.299) \\
0.8 & 0.226 & 0.007 & (0.213, 0.239) \\
0.9 & 0.139 & 0.005 & (0.129, 0.149) \\ \hline
\end{tabular}

\begin{tabular}{cccc}
\multicolumn{4}{c}{(b1) applying power-divergence}   \\
\hline
$\lambda$ & $\hat{V}^2_{H(f)}$ & SE & $95\%$CI \\  \hline
0.0 & 0.317 & 0.011 & (0.296, 0.339) \\
0.2 & 0.335 & 0.012 & (0.313, 0.358) \\
0.4 & 0.345 & 0.012 & (0.322, 0.369) \\
0.6 & 0.348 & 0.012 & (0.325, 0.372) \\
0.8 & 0.346 & 0.012 & (0.322, 0.370) \\
1.0 & 0.340 & 0.012 & (0.315, 0.364) \\
1.2 & 0.330 & 0.013 & (0.306, 0.355) \\
1.5 & 0.311 & 0.013 & (0.287, 0.336) \\ \hline
\end{tabular}

\begin{tabular}{cccc}
\\
\multicolumn{4}{c}{(2b) applying $\theta$-divergence}  \\
\hline
$\theta$ & $\hat{V}^2_{H(f)}$ & SE & $95\%$CI \\  \hline
0.0 & 0.340 & 0.012 & (0.315, 0.364) \\
0.1 & 0.364 & 0.012 & (0.340, 0.388) \\
0.2 & 0.372 & 0.012 & (0.348, 0.396) \\
0.3 & 0.370 & 0.012 & (0.346, 0.394) \\
0.4 & 0.359 & 0.012 & (0.335, 0.382) \\
0.5 & 0.339 & 0.012 & (0.316, 0.361) \\
0.6 & 0.309 & 0.011 & (0.287, 0.331) \\
0.7 & 0.268 & 0.011 & (0.247, 0.288) \\
0.8 & 0.210 & 0.009 & (0.191, 0.228) \\
0.9 & 0.127 & 0.007 & (0.114, 0.140) \\ \hline
\end{tabular}
\end{multicols}
\end{table}

\subsection*{Example 5}
Consider the data in Table \ref{eyes}, taken from \cite{tomizawa1985analysis}. 
These tables provide information on the unaided distance vision of 4746 university students aged 18 to about 25 and 3168 elementary students aged 6 to about 12. 
In Table \ref{eyes}, the row and column variables are the right and left eye grades, respectively, with the categories ordered from the highest grade (1) to the lowest grade (4).
These tables cannot distinguish whether the row and column variables are explanatory or response variables, so we apply the measure $V^2_{H(f)}$.

\begin{table}[hbtp]
\caption{Unaided distance vision data for university and elementary students}
\centering
\label{eyes}
\begin{tabular}{lccccc}
\multicolumn{6}{c}{(a) University Students} \\
\hline
 & \multicolumn{4}{c}{Left eye grade} & \\ \cline{2-5}
Right eye grade & (1) & (2) & (3) & (4) & Total  \\ \hline
Highest (1) & 1291 & 130 & 40 & 22 & 1483  \\
Second (2) & 149 & 221 & 114 & 23 & 507 \\
Third (3) & 64 & 124 & 660 & 185 & 1033 \\
Lowest (4)  & 20 & 25 & 249 & 1492 & 1723 \\   \hline
Total & 1524 & 500 & 1063 & 1659  & 4746 \\ \hline
\multicolumn{2}{l}{Source: \cite{tomizawa1985analysis}}  \\
\end{tabular}

\begin{tabular}{lccccc}
\\
\multicolumn{6}{c}{(b) Elementary Students} 
\\
\hline
 & \multicolumn{4}{c}{Left eye grade} & \\ \cline{2-5}
Right eye grade & (1) & (2) & (3) & (4) & Total  \\ \hline
Highest (1) & 2470 & 126 & 21 & 10 & 2627  \\
Second (2) & 96 & 138 & 33 & 5 & 272 \\
Third (3) & 10 & 42 & 75 & 15 & 142 \\
Lowest (4)  & 12 & 7 & 16 & 92 & 127 \\   \hline
Total & 2588 & 313 & 145 & 122  & 3168 \\ \hline
\multicolumn{2}{l}{Source: \cite{tomizawa1985analysis}}  \\
\end{tabular}
\end{table}

Table \ref{eyesr} gives the estimates of the measures, standard errors, and confidence intervals.
In addition, Tables \ref{eyesr}(a1, b1) and \ref{eyesr}(b1, b2) show the results of the analysis of Tables \ref{eyes}(a) and \ref{eyes}(b), respective.
The results of this analysis show that the two actual data have the strong structure of association in terms of the estimates and confidence intervals for all $\lambda$ and $\theta$.
As for the comparing value of the measures between Tables \ref{eyesr}(a1, b1) and \ref{eyesr}(b1, b2), it can be said that the strength of association between the right and the left eye is greater for elementary school students in term of the estimates for each parameter.
Comparing confidence intervals between Tables \ref{eyesr}(a1, b1) and \ref{eyesr}(b1, b2) can be similarly concluded, we can confirm that the confidence intervals for each $\lambda$ overlap, but not when $\theta = 0.5, 0.6, 0.7, 0.8, 0.9$.
Another interesting finding is that, unlike Example $2$, the results of the analysis in Example $5$ do not overlap the confidence intervals when $\theta = 0.5, 0.6, 0.7, 0.8$.
Including in terms of the triangular discrimination $\Delta$ ($\theta = 0.5$), which is not observed in the $V^2$ or the power-divergence type measures, it can be assumed that this result provides evidence that Table \ref{eyes}(b) has a stronger association structure.

\begin{table}[hbtp]
\caption{Estimate of the measure $V^2_{H(f)}$, estimated approximate standard error for $\hat{V}^2_{H(f)}$, and approximate $95\%$ confidence interval of $V^2_{H(f)}$ applying (a1, b1) power-divergence for any $\lambda$ and (a2, b2) $\theta$-divergence for any $\theta$.}
\label{eyesr}
\centering
\begin{multicols}{2}
\begin{tabular}{cccc}
\multicolumn{4}{c}{(a1) applying power-divergence}   \\
\hline
$\lambda$ & $\hat{V}^2_{H(f)}$ &  SE & $95\%$CI \\  \hline
0.0 & 0.413 & 0.017 & (0.379, 0.446) \\
0.2 & 0.420 & 0.017 & (0.386, 0.454) \\
0.4 & 0.415 & 0.018 & (0.380, 0.449) \\
0.6 & 0.400 & 0.018 & (0.364, 0.436) \\
0.8 & 0.380 & 0.019 & (0.342, 0.418) \\
1.0 & 0.357 & 0.020 & (0.317, 0.397) \\
1.2 & 0.334 & 0.022 & (0.291, 0.376) \\
1.5 & 0.300 & 0.023 & (0.255, 0.346) \\ \hline
\end{tabular}

\begin{tabular}{cccc}
\\ 
\multicolumn{4}{c}{(a2) applying $\theta$-divergence}  \\
\hline
$\theta$ & $\hat{V}^2_{H(f)}$ &  SE & $95\%$CI \\  \hline
0.0 & 0.357 & 0.020 & (0.317, 0.397) \\
0.1 & 0.448 & 0.017 & (0.414, 0.482) \\
0.2 & 0.474 & 0.017 & (0.440, 0.508) \\
0.3 & 0.479 & 0.018 & (0.444, 0.514) \\
0.4 & 0.471 & 0.018 & (0.435, 0.507) \\
0.5 & 0.453 & 0.019 & (0.416, 0.490) \\
0.6 & 0.425 & 0.019 & (0.387, 0.463) \\
0.7 & 0.384 & 0.020 & (0.345, 0.423) \\
0.8 & 0.323 & 0.020 & (0.283, 0.363) \\
0.9 & 0.223 & 0.019 & (0.185, 0.261) \\ \hline
\end{tabular}

\begin{tabular}{cccc}
\multicolumn{4}{c}{(b1) applying power-divergence}   \\
\hline
$\lambda$ & $\hat{V}^2_{H(f)}$ &  SE & $95\%$CI \\  \hline
0.0 & 0.461 & 0.009 & (0.443, 0.479) \\
0.2 & 0.465 & 0.009 & (0.447, 0.483) \\
0.4 & 0.459 & 0.009 & (0.441, 0.476) \\
0.6 & 0.445 & 0.009 & (0.428, 0.463) \\
0.8 & 0.426 & 0.009 & (0.408, 0.443) \\
1.0 & 0.402 & 0.009 & (0.384, 0.420) \\
1.2 & 0.375 & 0.009 & (0.357, 0.392) \\
1.5 & 0.330 & 0.009 & (0.312, 0.348) \\ \hline
\end{tabular}

\begin{tabular}{cccc}
\\
\multicolumn{4}{c}{(b2) applying $\theta$-divergence}  \\
\hline
$\theta$ & $\hat{V}^2_{H(f)}$ & SE & $95\%$CI \\  \hline
0.0 & 0.402 & 0.009 & (0.384, 0.420) \\
0.1 & 0.465 & 0.009 & (0.446, 0.483) \\
0.2 & 0.495 & 0.009 & (0.477, 0.513) \\
0.3 & 0.509 & 0.009 & (0.492, 0.527) \\
0.4 & 0.514 & 0.009 & (0.496, 0.532) \\
0.5 & 0.510 & 0.009 & (0.492, 0.528) \\
0.6 & 0.497 & 0.010 & (0.478, 0.515) \\
0.7 & 0.472 & 0.010 & (0.452, 0.492) \\
0.8 & 0.427 & 0.011 & (0.405, 0.449) \\
0.9 & 0.337 & 0.013 & (0.311, 0.363) \\ \hline
\end{tabular}
\end{multicols}
\end{table}

\clearpage
\bibliographystyle{apalike} 
\bibliography{sn-bibliography.bib}
\end{document}